\documentclass[11pt]{article}
\usepackage[margin=1in]{geometry} %to adjust page margins

\usepackage{graphicx}
\usepackage{multirow} %for multiple rows merged in table
\usepackage{array,booktabs,arydshln,xcolor}

\usepackage{listings}
\usepackage{enumerate}
\usepackage{float}
\usepackage{tikz}
\usetikzlibrary{arrows}
\usetikzlibrary{patterns.meta}
\usetikzlibrary{decorations.pathreplacing}
\usetikzlibrary {decorations.markings,fpu} 
\usetikzlibrary{decorations.pathmorphing}
\usepackage[shortlabels]{enumitem}

\usepackage{caption}
\usepackage{mathtools}%text above arrows
\usepackage{hyperref}%for links in refs
\usepackage{fancyhdr}%page style
\usepackage{relsize}

\usepackage{amsfonts}
\usepackage{amsmath}
\usepackage{bbm}
\usepackage{bm}
\usepackage{adjustbox}
\usepackage{ifthen}
\usepackage{amsthm}

\definecolor{toc}{RGB}{13,55,174}	%dark blue smth
\usepackage{hyperref}				%for links in refs
\hypersetup{
    colorlinks=true,
    citecolor=red,
    filecolor=black,
    linkcolor=toc,
    urlcolor=toc
}

\allowdisplaybreaks %to allow page break in equations
\usepackage{thmtools}
\usepackage{thm-restate}

\newtheorem{theorem}{Theorem}[section]

\newtheorem{lemma}[theorem]{Lemma}
\newtheorem{corollary}[theorem]{Corollary}

\newtheorem{fact}[theorem]{Fact}

\usepackage{cleveref}

\newcommand{\dist}{\mathcal{D}}
\newcommand{\lp}{\left}
\newcommand{\rp}{\right}
\newcommand{\E}[2]{\mathbb{E}_{#1}\lp[#2 \rp]}
\renewcommand{\Pr}[2]{\textbf{Pr}_{#1}\lp[#2 \rp]}

\newcommand{\e}{\varepsilon}

\newcommand{\Xs}{\mathbf{X}}
\newcommand{\Xsv}{\mathbf{X}_v}
\newcommand{\Xssv}{\mathbf{X}^*_v}
\newcommand{\Xsl}{\mathbf{X}^{+}}
\newcommand{\Xslv}{\mathbf{X}^{+}_v}
\newcommand{\Xl}{X^+}
\newcommand{\Qsv}{Q^{*}_v}
\newcommand{\Qsvo}{\overline{Q}^*_v}
\newcommand{\Qov}{Q^{o}_v}
\newcommand{\Qmax}{Q_{\max}}

\newcommand{\deltav}{\delta_v}
\newcommand{\Bv}{\mathbf{B}^{v}}
\newcommand{\pv}{p^v}
\newcommand{\qv}{q^v}
\newcommand{\Vs}{\mathbf{V}}
\newcommand{\Gv}{G_3^v}

\newcommand{\tp}{p^{o}}

\newtheorem{claim}{Claim} 
\newtheorem{definition}{Definition}

\newcommand{\gr}[1]{\textcolor{gray}{#1}}

\newcommand{\boxes}{\mathcal{B}}
\newcommand{\groups}{\mathcal{G}}
\newcommand{\opt}{\text{OPT}}
\newcommand{\alg}{\text{ALG}}
\newcommand{\poly}{\text{poly}}
\newcommand{\polylog}{\text{polylog}}
\newcommand{\supp}{\text{supp}}

\newcommand{\remove}[1]{}

\newcommand{\eps}{\varepsilon}
\usepackage{natbib}
\date{} 

\title{Prophet Secretary Against the Online Optimal}

\author{
Paul Dütting \\ Google Research \\ {\tt duetting@google.com} \and 
Evangelia Gergatsouli \\ University of Wisconsin-Madison \\ {\tt evagerg@cs.wisc.edu} \and
Rojin Rezvan \\ University of Texas at Austin \\ {\tt rojin@cs.utexas.edu} \and
Yifeng Teng \\ Google Research \\ {\tt yifengt@google.com} \and 
Alexandros Tsigonias-Dimitriadis \\ Universidad de Chile \\ {\tt alexandrostsigdim@gmail.com}
}

\newtheorem*{remark}{Remark}

\begin{document}

\maketitle

We study the prophet secretary problem, a well-studied variant of the classic prophet inequality, where values are drawn from independent known distributions but arrive in uniformly random order. Upon seeing a value at each step, the decision-maker has to either select it and stop or irrevocably discard it. Traditionally, the chosen benchmark is the expected reward of the prophet, who knows all the values in advance and can always select the maximum one.
In this work, we study the prophet secretary problem against a less pessimistic but equally well-motivated benchmark; the \emph{online} optimal. Here, the main goal is to find polynomial-time algorithms that guarantee near-optimal expected reward. As a warm-up, we present a quasi-polynomial time approximation scheme (QPTAS) achieving a $(1-\e)$-approximation in $O(n^{\text{poly} \log n\cdot f(\e)})$ time through careful discretization and non-trivial bundling processes. Using the toolbox developed for the QPTAS, coupled with a novel \emph{frontloading} technique that enables us to reduce the number of decisions we need to make, we are able to remove the dependence on $n$ in the exponent and obtain a polynomial time approximation scheme (PTAS) for this problem.

\setcounter{page}{0}
\thispagestyle{empty}
\newpage

\section{Introduction}
The prophet inequality problem is a central problem in the field of optimal stopping theory. In the classic version of this problem, a gambler faces a sequence of $n$ variables $X_i$, each drawn independently from a known distribution $\dist_i$, which are presented to him one by one.  The gambler's goal is to decide, at each step, whether to keep the value realized from the distribution or continue and irrevocably discard it, in order to maximize the value obtained. The optimal online policy can be found (in poly-time) using backward induction. The classic benchmark for the classic problem is the all-knowing \emph{prophet}, who can see the future values and select the maximum. %%
The celebrated result of \cite{KrenSuch1977,KrenSuch1978,Samu1984} shows that using a simple threshold policy, the gambler can obtain an expected value that is at least 1/2 of the expected value of the prophet, and that this is best possible among all online policies. Prophet inequalities are of special interest in algorithmic game theory too, due to their strong connections to sequential posted pricing mechanisms. 
In fact, we know from \cite{HajiKleiSand2007,ChawHartMaleSiva2010, CorrFoncPizaVerd2019} that designing posted-price mechanisms is equivalent to finding (threshold) stopping rules in the related prophet inequality setting. 

Motivated in part by this connection to sequential posted pricing, prophet inequalities have been studied for a broad range of combinatorial settings such as matroids
\citep{KleiWein2012,Alae2011,FeldmanSZ16}, bipartite and non-bipartite matching with vertex or edge arrivals \citep{GravinW19,EzraFGT20,EzraFGT22}, or combinatorial auctions \citep{FeldmanGL15,DuetFeldKessLuci2017,DuttingKL20}.

A fundamental variant of the classic prophet inequality problem is the prophet secretary problem~\citep{EsfaHajiLiagMone2015, CorrFoncHoekOostVred2017, AzarChipKapl2018,CorrSaonZili2019}, where the variables arrive in uniformly random order instead of adversarially. 
In the original formulation of this problem, we compete again with the \emph{prophet}. 
In this case, it was shown that there exists a $0.669$-approximation \citep{CorrSaonZili2019} while no algorithm can achieve an approximation better than $0.7254$~\citep{BubnaChipl2022}. 
Closing this gap is an important open problem in the prophet inequalities community. Another exciting property of the prophet secretary problem is that unlike in the classic prophet inequality problem, the dynamic program for solving the optimal online policy is of exponential size. 

In this work, we study the prophet secretary problem against the \emph{online optimal}.
This benchmark does not assume any prior knowledge of the future; we compete against an algorithm that has the same information as we do at every step, but infinite computation power. This way, we measure the potential loss that arises due to computational limitations, rather than quantifying the loss that's due to the fact that the algorithm has to make decisions online. 
Our work is the first to provide algorithms for prophet secretary against the online optimal. We design both a quasi-polynomial time approximation scheme (QPTAS) and a polynomial time approximation scheme (PTAS), showing that we can obtain an arbitrarily good approximation against the online optimal. Our algorithmic results uncover important structural properties of the prophet secretary problem and near-optimal (threshold-based) stopping decisions.

\subsection{Our Contribution \& Techniques}

We show how to approximate the expected value of the online optimal to within a factor of $1-\e$; this immediately translates into an algorithm that achieves a $1-\e$ approximation.

Our starting point is a simple observation, which states that if one groups ``similar'' variables into $g$ groups and treats variables in each group in the same way, then one can write a dynamic program that tracks the number of variables in each group, with complexity exponential in $g$. This reduces the problem to showing that there is a succinct grouping that obtains a $(1-\e)$-approximation. For the QPTAS it suffices to give a grouping of size $\polylog(n)$, for the PTAS we need to further reduce this to $O(1)$.

\smallskip
\noindent
\textbf{First result: QPTAS.}
We give the intuition underlying the grouping of size $\polylog(n)$ for a special case that captures most of the solution idea. The special case is that each random variable has binary support as follows:  variable $X_i$ is either $v_i$ or zero, with probabilities $p_i$ and $1-p_i$. 

First, we argue that we can scale the values appropriately so that OPT falls into some small constant interval $[c,1]$. Afterward, we show that we don't lose more than $O(\e)$ if we ignore variables with low value ($v_i \leq \e$) or low expected value ($v_i p_i \leq \e/n$). 
For the remaining variables whose value is very high ($v_i \geq 1/\e^2$), we show how to compress them by adjusting their values so that they all have the same value and their probabilities fall in a $O(\poly~n)$ range. We then discretize the variables with small values (from $\e$ to $1/\e^2$) and their respective probabilities which range from $\e^3/n$ to $1$, to powers of $(1+\e)$. We also round the probabilities (which are also in a $\poly(n,\e)$ range) of the variables with the same high value to powers of $(1+\e)$. 

This construction readily generalizes to the constant-size support case (by treating the O(1) high realizations as we did before for the single one).
An additional ingredient that is needed for the case where the support is not necessarily of constant size, is an argument for collapsing all high realizations (those above OPT) of a variable into a single point. Together with a discretization of the low realizations of a variable (from $\e$ to 1), we are back to the constant-size support case.

\begin{theorem}[QPTAS, informal]
There exists a $(1-\e)$-approximation algorithm for computing the optimal reward of the prophet secretary problem in time 
$n^{\polylog~n \cdot f(\e)}$.
\end{theorem}

\smallskip
\noindent
\textbf{Second result: PTAS.}
Similar to the case of the QPTAS, the case where all variables are binary is descriptive of the idea behind the PTAS for the general case. Starting off from the discretizations done in the QPTAS, what we need to take care of are variables
with realization probability that are each individually small ($<poly(\epsilon))$, but that still cannot be ignored as there may be enough of them to contribute considerably to optimal reward. This step is crucial to reduce the number of possible probabilities to $\poly(\frac{1}{\eps})$.
% \atd{is what we mean here the following? We care for both the small and the very small, but the small are the ones that are difficult to handle becuase they still have a small constant probability. If yes, then we should rephrase. Now it sounds like we do not care or we can throw away the too small realization probabilities. Same comment applies in p.12.}

In order to do so, we propose the novel technique of ``frontloading". The idea is as follows. Fix a support value $v$, and consider the variables of interest (those with neither high nor low probabilities) with that support. If it is the case that $k$ many of these variables have a total realization probability that is not very high, then we claim that imagining these $k$ variables as a single box, with the total realization probability equal to that of the $k$ boxes, and as an "outside option" always available through the interval where these variables arrive, does not affect the reward much. This makes us capable of reducing the number of different probabilities we need to track for the dynamic program that finds the optimal. 

In order to generalize this to variables with support size more than one, we show that we can imagine a variable with multiple support values, as multiple binary variables (each of them having one of the support values with its corresponding probability that is not too high, and the remaining probability on value $0$) and an additional variable with high probability for each support value that arrives after each other. This brings us back to the binary case, as there are not many types of variables with a larger probability on each support value and they can be treated differently. 

\begin{theorem}[PTAS, informal]
There exists a $(1-\e)$-approximation algorithm for computing the optimal reward of the prophet secretary problem in time 
$n^{f(\e)}$.
\end{theorem}

The above theorem shows that the optimal reward $\opt$ of any prophet secretary instance can be efficiently approximated. Observe that given an efficient oracle that always returns a value between $(1-\eps)\opt$ and $\opt$, we can use it to construct an efficient algorithm for the prophet secretary problem with an expected reward at least $(1-\eps)\opt$. In particular, at each step, the algorithm queries the oracle on the remaining instance to obtain a threshold to make the decision on the current variable. See Lemma~\ref{lem:error_propagation} for more discussion. Then we not only have a PTAS for computing the approximately optimal reward of the prophet secretary problem, but also obtain an efficient strategy with almost optimal reward.

\begin{corollary}[PTAS Algorithm, informal]
There exists a $(1-\e)$-approximation strategy for the prophet secretary problem against the online optimal that runs in time 
$n^{f(\e)}$.
\end{corollary}

\subsection{Related Work}
\smallskip
\noindent
\textbf{Classic prophet inequality and prophet secretary.} There is a vast literature on the topic of prophet inequalities, starting with the classic, fixed order prophet inequality studied first in \citet{KrenSuch1977,KrenSuch1978}, who obtained a tight 0.5 approximation for the problem. Another remarkable result came a bit later from \citet{Samu1984}, who showed that the tight bound can be achieved by using a single-threshold algorithm (a different well-known single-threshold algorithm was also shown in \citet{KleiWein2012}). For a comprehensive treatment of the subject, we refer the reader to the earlier survey by \citet{HillKert1992} and the more recent ones by \citet{Luci2017, CorrFoncHoekOostVred2018}. 

In this paper, we focus on one of the most well-established variants, which is a natural combination of the secretary problem and the classic prophet inequality, introduced by \citet{EsfaHajiLiagMone2015} as the prophet secretary problem. Their main result was to show that a non-adaptive multi-threshold algorithm achieves a $\lp( 1 - \frac{1}{e} \rp)$-approximation to the prophet benchmark. A bit later, \citet{EhsaHajiKessSing2018} and \citet{CorrFoncHoekOostVred2017} arrived to the same approximation guarantee by using different threshold-based algorithms. The first algorithm breaking this $\lp( 1 - \frac{1}{e} \rp)$ barrier by a tiny fraction came from \citet{AzarChipKapl2018} and was subsequently improved by \citet{CorrSaonZili2019} to $0.669$, which is the currently best-known bound. Obtaining a tight answer for this problem remains a very interesting open question, since the best-known hardness result is $0.7254$, as shown very recently in \citet{BubnaChipl2022}.

The exploration of the optimal online algorithm as a benchmark has been pioneered in \citet{NiazSabeSham2018} and \citet{PapadPollSabeWajc2021, BravDerakMol2022} as well as \citet{AnarNiazSabSham2019}.  \citet{NiazSabeSham2018} study the power of single-threshold algorithms when competing with the online optimal. \citet{PapadPollSabeWajc2021}, \citet{BravDerakMol2022}, and \citet{NaorSrinWajc2023} consider this benchmark for the online stochastic maximum-weight matching problem under vertex arrivals,\footnote{Note that the online rounding scheme for fractional bipartite matchings of \cite{SabWajc2021} also improved and generalized the result of \cite{PapadPollSabeWajc2021}.} and \citet{AnarNiazSabSham2019} for an online pricing problem with laminar matroid constraints.

\smallskip
\noindent
\textbf{The free-order model.} Apart from the fixed and the random arrival order models, another setting that has received a lot of attention is the one where the decision-maker can freely choose the order. For the optimal ordering problem, where the goal of the decision-maker is to compute (or approximate) the optimal order of inspection using efficient algorithms, \citet{AgraSethZhan2020} showed that it is NP-hard and provided an FPTAS for distributions of support size $3$. \citet{Fu2018} proved, among other results, that for distributions with constant support size, the problem admits a PTAS. Earlier, \citet{ChakEvenGuhMansMuth2010} gave a PTAS for general distributions for the problem of revenue maximization in sequential posted-price mechanisms; combined with the result of \citet{CorrFoncPizaVerd2019}, they automatically obtain a PTAS for the optimal ordering problem as well. Finally, \citet{LiuLemeSchneiPalBala2021} designed an EPTAS using a new decomposition technique for random variables. An EPTAS was also independently obtained by \citet{SegeSing2021}. A remaining open question is whether there can be an FPTAS.

When competing with the prophet benchmark, called the order selection problem, until very recently the best-known bound was $0.669$ coming from the work of \citet{CorrSaonZili2019} for the random order model. Very recently, \citet{PengTang2022} was the first to beat this bound and obtain a (much-improved) 0.7251-competitive algorithm, which was further improved to 0.7258 by \citet{BubnaChipl2022} shortly thereafter. Since no separation is known between the i.i.d. prophet inequality and the order selection problem, the upper bound of $0.745$ follows from the work of \citet{HillKertz1982}. Showing or disproving that the i.i.d.~case is the worst-case for free-order remains an intriguing open problem.

\smallskip
\noindent
\textbf{The Pandora's Box problem.} A similar setting, where there is additional exploration cost for the variables but the decision-maker does not have to make immediate and irrevocable decisions, is the Pandora's box problem, first defined more than four decades ago by~\citet{Weit1979}. A recent line of work extended the original model in various directions (see, e.g., ~\cite{GuhMunSar2008, Dova2018,BeyhKlei2019,BoodFuscLazoLeon2020,ChawGergTengTzamZhan2020,ChawGergMcmaTzam2021}), where in all cases the goal is either to maximize or minimize the objective function against the online optimal. 

In particular, the setting of non-obligatory inspection has been thoroughly explored from an algorithmic point of view, and the question of whether there exists any polynomial-time algorithm to compute the optimal policy remained open until very recently. \citet{GuhMunSar2008} provided the first constant-factor approximation to the problem. Note that their paper deals with information acquisition in multichannel wireless networks, but their model turns out to correspond exactly to the Pandora's box problem with non-obligatory inspection. \citet{Dova2018} introduced the problem in the economics literature some years later. \citet{BeyhKlei2019} revived it in the EconCS community; their main result was to introduce a simple policy that also guarantees a constant-factor approximation (albeit worse than the one of \citet{GuhMunSar2008}). Only very recently \citet{FuLiLiu2022} showed that the problem is NP-hard and provided a PTAS. In a concurrent and independent work, \citet{BeyhCai2022} also obtained a PTAS for the problem. 

Finally, in the interesting problem termed Pandora's box with commitment, in which the decision-maker now has to decide immediately and irrevocably whether to stop or continue, \citet{Fu2018} first gave a PTAS and later \citet{SegeSing2021} developed an EPTAS for this problem, as part of a general framework that works in different classes of stochastic combinatorial optimization problems.

\paragraph{Organization. }In Section~\ref{sec:prelims}, we define the optimal objective and state a dynamic program to compute it, together with some structural lemmas that compare the optimal reward of instances with similar variables. In Section ~\ref{sec:qptas}, we first describe the QPTAS strategy for a special case of two-point distributions and later extend it to the general case. Finally, in Section~\ref{sec:ptas}, we describe the PTAS strategy again for two-point distributions. In the interest of space, the proofs for the general case of PTAS have been moved to the Appendix. Omitted proofs of Sections \ref{sec:prelims}, \ref{sec:qptas}, and \ref{sec:ptas} can be found in Appendices \ref{apn:prelims}, \ref{apn:qptas}, and \ref{apn:ptas}, respectively.

\section{Preliminaries}\label{sec:prelims}

Let $X_1, \ldots, X_n$ be $n$ independent non-negative random variables drawn from known distributions $\dist_1, \ldots, \dist_n$. To simplify the presentation, we focus on the case of discrete distributions. We remark that all of our results hold also for continuous distributions; we discuss this further at the corresponding technical sections.
We denote the product distribution by $\dist = \dist_1 \times \ldots \dist_n$ and the support of each variable $X_i$ by $\supp(X_i)$. We are sequentially presented with the variables in a uniformly random order and need to decide at every step to either keep the value realized from the distribution, or irrevocably discard it. The goal of the algorithm is to maximize the expected value chosen. 

As a benchmark we adopt the \emph{online optimal}, i.e., the expected value achievable by an algorithm that has infinite computational resources but no knowledge of the future. 

The online optimal can be expressed as a dynamic program, as shown in equation~\eqref{eq:DP}. We denote by $\opt(X| \boxes)$ the recursive solution when the current variable has value $X$ and the set of variables remaining is $\boxes$.
\begin{equation}\label{eq:DP}
\opt( X| \boxes)  = \begin{cases}
\E{\dist}{X}& \text{if } \boxes = \emptyset \\
\E{\dist}{\max\lp( X , \frac{1}{|\boxes|}\sum_{i\in\boxes} \opt( X_{i}| \boxes\setminus \{i\} ) \rp)}& \text{else.} \end{cases}
\end{equation}

Observe that the optimal strategy is a series of (adaptive and decreasing) thresholds; in every step we only continue if the current value obtained is less than the expected optimal of the subproblem, with one less box. We denote the threshold of $|\boxes| = n$ variables by 
\[ 
\theta_n = \frac{1}{|\boxes|}\sum_{i\in\boxes} \opt( X_{i}| \boxes\setminus \{i\} ).
\]
 We also denote by $X_{-i}$ all variables $X_1, \ldots, X_n$ except the variable $X_i$ and by $Y_{k:n}$ (resp. $X_{k:n})$ for all $k \in [n]$ the remaining instance starting at index $k$. For the instance we want to solve, i.e. $Y_{1:n}$ (resp. $X_{1:n}$), we drop the subscript and simply write $Y$ (resp. $X$). Similarly, $\opt_{\ell:k}$ is the optimal solution for the sub-instance $X_\ell, \ldots X_k$. 
 
Our ultimate goal is to find a poly-time algorithm that achieves a $1-\e$ approximation to the online optimal.  Following, we give the definitions of QPTAS and PTAS in the context of maximization problems.

 \begin{definition}[QPTAS]
For every fixed $\e>0$, there is a $(1-\e)$-approximation algorithm that runs in time $O\lp(n^{\polylog\ n} \rp)$.
 \end{definition}

 \begin{definition}[PTAS]
 For every fixed $\e>0$, there is a $(1-\e)$-approximation algorithm that runs in time polynomial in $n$.   
 \end{definition}

\subsection{From reward approximation oracle to efficient strategy}
Lemma~\ref{lem:error_propagation} shows that if we can calculate thresholds that are always within $(1-\e)$ of those that the DP calculates, then the error does not propagate and we will lose in expectation at most $(1-\e)$. The proof appears in Section~\ref{apn:prelims} of the Appendix. The lemma can transform any efficient approximation algorithm on the optimal reward of the prophet secretary algorithm to an efficient strategy for the prophet secretary game with approximately optimal reward.

\begin{restatable}{lemma}{errorProp}\label{lem:error_propagation}
For any prophet secretary instance $X_1,\cdots,X_n$, let $\opt'(X)$ be an algorithm such that given input $X$ returns a value $\opt'(X)\in[(1-\eps)\opt(X),\opt(X)]$. Consider an algorithm $\alg$ for the prophet secretary problem such that at any step, the algorithm uses $\opt'($on the remaining variables$)$ as a threshold for the current variable. Then 
\[\alg(X)\geq (1-\eps)\opt(X).\]
\end{restatable}

\subsection{Grouping in the DP}
Since the dynamic program given above can be exponential in size, a general tool that we will employ is to group variables into a smaller number of groups, and treat all variables in a group in the same way. 

Specifically, given a collection of $g$ groups of variables,
we only need to keep track of the number of variables in each group. Denote by $k_i$ the size of group $i$, and by $K=\sum_i k_i$ the total number of groups remaining. Then we can write the optimal DP as follows

\begin{equation}\label{eq:groups}
\opt( X, k_1,\ldots, k_g)  = 
\begin{cases}
\E{\dist}{X}& \text{if } K=0, \\
\E{\dist}{\max\lp( X, \frac{1}{K}\sum\limits_{i\in [g]} k_i \opt( X_{i}, k_1,\ldots,k_{i-1},k_i\text{-}1,k_{i+1}\ldots,k_g ) \rp)}& \text{else.} \end{cases}
\end{equation}

A simple observation now is that the size of the optimal online, using formulation~\eqref{eq:groups}, is exponential in the number of different groups. The proof of this claim is deferred to Section~\ref{apn:prelims} of the appendix.

\begin{restatable}{claim}{clGroups}\label{cl:DP_group_size}
The size of the dynamic program of ~\eqref{eq:groups}, given $g$ different non-empty groups, is at most $\left\lceil \frac{n}{g} \right\rceil^g$, where $n$ is the number of variables.
\end{restatable}

With this observation at hand, it suffices to find a grouping of the variables into $\polylog~n$ groups for a QPTAS and into a constant number of groups for a PTAS, such that the online optimal on the groups is $1-\e$ close to the online optimal.

\subsection{Structural lemmas}
We next present two structural lemmas, which will be useful in our analysis. Lemma~\ref{lem:value_error_propagation} shows that if we perturb the values that the random variables can take a bit, then the solution to the DP does not change by much. Lemma~\ref{lem:no_propagation_probability_error} shows that this is also true for probabilities. Their proof is deferred to section~\ref{apn:prelims} of the Appendix.

\begin{restatable}{lemma}{errorPropValue}\label{lem:value_error_propagation}
Let $\opt_{1:n}$ be the optimal online solution to prophet secretary on variables $X_1, \ldots X_n$, and $\opt'_{1:n}$ the optimal online solution on an instance where each value $v\in \supp(X_i)$ is replaced by $v'$ s.t. $v' \leq v \leq \gamma\cdot v'$ then it holds that
\[
\opt'_{1:n} \geq \frac{\opt_{1:n}}{\gamma}.
\]
Similarly, if $v' \leq v \leq v'+\gamma$ for some $\gamma>1$ then $\opt'_{1:n} \geq \opt_{1:n} - \gamma$.
\end{restatable}

\begin{restatable}{lemma}{errorProb}\label{lem:no_propagation_probability_error}
Let $\opt_{1:n}$ be the optimal online solution to prophet secretary on variables $X_1, \ldots X_n$, and $\opt'_{1:n}$ the optimal online solution on an instance where each probability $p_v=\Pr{\dist_i}{X_i=v}$ for $v\in \text{supp}(X_i)$ is replaced by $p_v'$ s.t. $p_v' \leq p_v \leq \gamma\cdot p_v'$ then it holds that
\[
\opt'_{1:n} \geq \frac{\opt_{1:n}}{\gamma}.
\]
\end{restatable}

\section{Warm-up: QPTAS}\label{sec:qptas}
We begin by presenting a QPTAS for approximating the optimal reward of the prophet secretary problem. By Claim~\ref{cl:DP_group_size}, our goal is to show that there is a grouping of the variables into $\polylog(n)$ groups such that the online optimal on the groups achieves a ($1-\e$)-approximation. 

The key insights that will enable this are: for each variable, we can collapse values and probabilities of all realizations above OPT into a single point. Moreover, we can compress the support of the probabilities of these points into a $\poly(n)$ range by moving all their values to a single value that we choose. 

\begin{remark}
Very simple ideas, like grouping according to the mean and variance of each distribution fail. For more details see Section~\ref{apn:counterexample} of the Appendix. 
\end{remark}

\subsection{QPTAS for constant-size support}\label{subsec:qptas_2_point}
We start by solving the case where each random variable comes from a distribution with constant support size, as this already captures some of the key ideas of the construction. Any proof not included in this section is deferred to Section~\ref{apn:qptas_2_points} of the Appendix.

A variable $X_i$ (with with $\E{}{X_i} < \infty$) has support size at most $c$, for some constant $c$, if it is of the form
\[X_i = \begin{cases}
v_1 & \text{w.p. } p_1\\
v_2 & \text{w.p. } p_2\\
\ldots & \\
v_\kappa & \text{w.p. } p_\kappa,\\
\end{cases}\]
where $\kappa \leq c$ is the size of the support. Next, we state the main result of this section.

\begin{restatable}{theorem}{qptasMain}[QPTAS for constant size support]\label{thm:QPTAS_2_point} 
There exists a $(1-\e)$-approximation algorithm for the prophet secretary problem against the online optimal that runs in time $O\lp( \lp(\frac{n}{g}\rp)^g \rp)$ for $g =O\lp(\frac{\log^{2c} n/\e}{(c-1)! \log^{2c} (1+\e)}\rp)$. In the special case of $c=2$, for 2-point distributions with $v_1 = 0$ and $v_2 > 0$ we can get running time of
% Full expression, kinda ugly looking : $O\lp(\lp( \frac{n \log^2 (1+\e)}{\log^2 (n/\e)}\rp)^{\frac{\log^2 n/\e}{\log^2 (1+\e)}}\rp)$ 
$O\lp( \lp(\frac{n}{g}\rp)^g \rp)$ for $g =O\lp(\frac{\log^2 n/\e}{\log^2 (1+\e)}\rp)$. 
\end{restatable}

In what follows, we give a proof for the special case of two-point distributions, each supported on one zero and one non-zero value, as it contains most of the important ideas that we need to prove the theorem. At the end of this subsection, we describe the small extra step we need to do to obtain the result for the constant-size support case, and we defer the rest of the necessary adjustments to the appendix. After this, we will move on to the case of general distributions. Thus, from now on we consider random variables $X_i$ of the form 
\[X_i = 
\begin{cases}
v_i & \text{, w.p. } p_i\\
0 & \text{, w.p. } 1-p_i.
\end{cases}
\]

\bigskip
\noindent\textbf{Preprocessing:} Before describing our discretization process, we do some preprocessing steps, starting by scaling the optimal such that the following claim holds.

\begin{claim}\label{cl:bounded_opt_off}
WLOG we can normalize the values s.t. $\opt \in [0.669,1]$.
\end{claim}
This holds since we can normalize the values in order to have $\E{\dist}{\max_i X_i}=1$. Using Theorem 1.1 from~\cite{CorrSaonZili2019} we know that the online $\opt$ will be in $[0.669,1]$ (since they give a $0.669$ approximation to the prophet). Following this, for the variables where $v_i\leq \e$ or $v_i p_i \leq \e/n$ we set  $v_i'=0$ and keep $p_i'=p_i$. 

The following claim shows that grounding to $0$ the point masses $v_i$ with $v_i \cdot p_i \leq \frac{\e}{n}$ does not incur more that $O(\e)$ loss for $\opt$.

\begin{restatable}{claim}{vanishingGain}
\label{cl:small_mean_vanishing_gain}
Let $X_1, \ldots,X_n$ be two-point random variables. If for every variable $X_i$ for which $\E{}{X_i} = v_i \cdot p_i \leq \frac{\e}{n}$ we move all mass to $0$, (i.e., $v'_i = 0$ with $p'_i=1$)  then we lose at most $1.5\e \cdot \opt$.
\end{restatable}

\smallskip
\noindent\textbf{Dealing with high values:} Before proceeding to the discretization, and since we did not make any assumptions on the upper bound of the support of every distribution, we need to treat separately the points with ``high'' values. Here by ``high'' values we mean those variables that have a point with $v_i\in [1/\e^2, \infty)$ and $p_i \in (0,\e^2]$. In particular, we need to compress their values or their probabilities in a $\poly(n)$-range support (in values / probability) so that we can perform afterwards the careful discretization process. We achieve that by performing the following expectation-preserving transformation. Let $v_{\max}$ be the largest support value in the input. We set all such points to have equal value $v_i'=v_{\max}$ and rescale the initial probabilities to $p_i' \leftarrow \frac{v_ip_i}{v_{\max}}$, to maintain the same expected reward from this support value (i.e., $v_i' p'_i = v_i p_i$). Let $\tp$ be the largest $p'_i$ from this transformation, we can ensure that $p_i'\in [\frac{\e}{n}\tp, \tp]$ for each $i\in \boxes$ by removing all variables with $p'_i<\frac{\eps}{n}$ losing at most $O(\eps)$ fraction of the optimal reward due to Claim~\ref{cl:small_mean_vanishing_gain}. 

We claim that this transformation does not cause more than $O(\eps^2)$ multiplicative loss. 

\begin{restatable}{claim}{constloss}\label{cl:single-large-value}
    Let $X_1, \ldots, X_n$ be an instance of prophet secretary such that each variable $X_i$ has at most one support value $>1$, and $Y_1, \ldots Y_n$ the instance where we change the variables $X_i$ to variable $Y_i$, where the support value $v_i>1/\eps^2$ with probability $p_i$ is changed to $v_{\max}$ with probability $\frac{v_i p_i}{v_{\max}}$, where $v_{\max}$ is the largest support value of all variables. Then
    \[
    \opt(X) \leq \opt(Y) \leq (1+\eps^2)\opt(X).
    \]
\end{restatable}

The proofs of Claim~\ref{cl:single-large-value} and Claim~\ref{cl:small_mean_vanishing_gain} are deferred to Appendix~\ref{apn:qptas_2_points}.
\bigskip
\noindent\\ \textbf{Discretization:}  
We now define the steps in our discretization process.
\begin{definition}[Discretization]\label{def:discretization}
Given a set of $n$ variables, each of which is a two point distribution giving $v_i$ w.p. $p_i$ and $0$ otherwise we define the following discretization process.
\begin{itemize}
    \item \textbf{Step 1}: round down all $v_i$'s to the nearest $(1+\e)^k$, for $k\in \mathbb{Z}$. 
    \item \textbf{Step 2}: round down all $p_i$'s to the nearest $(1+\e)^k$, for $k\in \mathbb{Z}$.
\end{itemize}
\end{definition}

We discretize each group of variables differently, based on the values of $v_i$ and $p_i$ of each variable. The different discretization cases are also presented in Figure~\ref{fig:cases}.

\begin{itemize}
    \item \textbf{Case 1}: variables that have $v_i \in[\e,1/\e^2]$ and $p_i \in[\e^3/n,1]$. We use the discretization of Definition~\ref{def:discretization}, based on both value and probability.

    \item \textbf{Case 2}: variables that have $v_i\in [1/\e^2, \infty)$ and $p_i \in (0,\e^2]$. Here, we have performed the expectation-preserving transformation described above and then we use again the discretization of Definition~\ref{def:discretization}, based on both value and probability.
\end{itemize}

Observe also that variables that have $v_i > 1/\e^2$ and $p_i \in (\e^2,1]$ cannot exist, since in this case the optimal would not be $\opt\leq 1$.
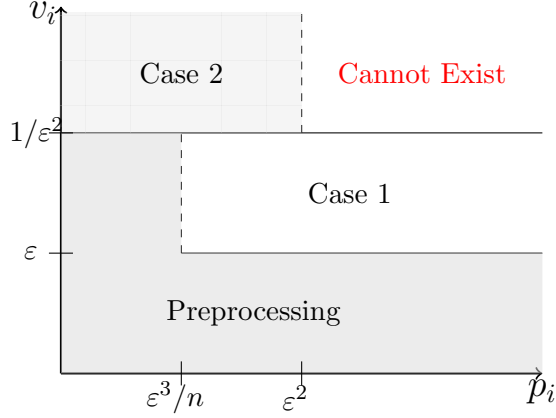
\begin{figure}[H]
    \centering
    \begin{tikzpicture}[scale=0.8]

\pgfmathtruncatemacro{\casesL}{2};
\pgfmathtruncatemacro{\totalL}{8};

% \tikzset{secondCut/.style={dashed}}
% \tikzset{thirdCut/.style={dotted, thick}}
\tikzset{mycolor/.style={fill=gray!30, opacity=.5}}

% \tikzset{mycolorStripped/.style={pattern={Lines[angle=-45, distance=2mm, line width=1.1mm, xshift=1.4mm]},  pattern color=orange!50!yellow}}

%Draw horizontal lines
\draw[->, thick] (0,0)--(0,\totalL-1.9) node[]{};
\draw[-] (2,2)--(\totalL,2) node[]{};
\draw[-] (0,4)--(\totalL,4) node[]{};

%Draw vertical lines
\draw[-, dashed] (\casesL,2)--(\casesL,4) node[]{};
\draw[->, thick] (0,0)--(\totalL,0) node[]{};
\draw[-, dashed] (2*\casesL,4)--(2*\casesL,6) node[]{};

%Draw mlkies inside triangle

%Fill shapes with color
\fill[mycolor] (0,0)--(0,4)--(2,4)--(2,2)--(\totalL,2)--(\totalL,0);
\fill[pattern={Lines[angle=45, distance=6mm, line width=13mm, xshift=5mm]}, pattern color=gray!30,opacity=.5] (0,4)--(0,6)--(4,6)--(4,4);

\node[] at (-0.5,2) {$\varepsilon$};
		\draw[-] (-0.2, 2) --(0.2,2);

		% OLD VERSION \node[] at (-0.4,4) {$n^3$};
  		\node[] at (-0.4,4) {$1/\e^2$};
		\draw[-] (-0.2, 4) --(0.2,4);

\node[] at (-0.3,\totalL-2) {{\Large $v_i$}};
\node[] at (\totalL,-0.3) {{\Large $p_i$}};
		
\draw[-] (\casesL, -0.2) -- (\casesL, 0.2);
% OLD VERSION \node[] at (\casesL-0.1,-0.3) {$\frac{\e}{n^4}$};
\node[] at (\casesL-0.1,-0.4) {$\e^3/n$};

\draw[-] (2*\casesL, -0.2) -- (2*\casesL, 0.2);

% OLD VERSION  \node[] at (2*\casesL-0.1,-0.3) {$\frac{1}{n^3}$};
\node[] at (2*\casesL-0.1,-0.4) {$\e^2$};
%Cases
\node[] at (0.4*\totalL,1) {Preprocessing};
\node[] at (0.12*\totalL,3) {};
\node[] at (0.6*\totalL,3) {Case 1};
\node[] at (0.25*\totalL,5) {Case 2};
\node[] at (0.75*\totalL,5) {\color{red}{Cannot Exist}};

\end{tikzpicture}
    \caption{Different cases for discretization for the $2$-point case. In case 2, we only consider the variables with $v_ip_i>\e/n$. Preprocessing here refers to all the value-probability pairs that we omit from the instance.}
    \label{fig:cases}
\end{figure}

\paragraph{Proof of Theorem~\ref{thm:QPTAS_2_point} (for 2-point distributions).} 
We separately bound the error and calculate the number of groups resulting from our discretization process.

\medskip
\noindent\emph{Bounding the error:}
In our preprocessing phase, we move all mass to value $0$ for the following two types of variables: when  $v_i<\e$ (step 1), and when $v_i p_i <\e/n$ (step 2). Denote by $\opt'$ and $\opt''$ the  optimal value after each of the preprocessing steps.
\begin{itemize}
\item \textbf{Variables with $v_i<\e$}: note that each of the variables in Step 1 contributes at most $\e$ and $\opt\leq 1$ from Assumption~\ref{cl:bounded_opt_off}. Denote by $A = \{ X_i : v_i > \e\}$ and we write $\opt$ as
    \begin{align*}
    \opt & = 
    \E{}{\text{gain}|\text{select }b\in A}\Pr{}{\text{select }b\in A} +
        \E{}{\text{gain}| \text{select }b\not \in A} \Pr{}{\text{select } b\not\in A}\\
        & \leq     \E{}{\text{gain}|\text{select }b\in A}\Pr{}{\text{select }b\in A}  + \frac{\e}{0.669} \opt\\
        & \leq \opt' + 1.5\e \opt,
    \end{align*}
    from which we get $\opt' \geq (1-1.5\e)\opt$.
    \item \textbf{Variables with }$v_i p_i < \e/n$: denote by $A = \{ X_i : v_i p_i \geq \e/n\}$, then we write $\opt'$ as
    \begin{align*}
    \opt' & = 
    \E{}{\text{gain}|\text{select }b\in A}\Pr{}{\text{select }b\in A} +
        \E{}{\text{gain}| \text{select }b\not \in A} \Pr{}{\text{select } b\not\in A}\\
        & \leq \E{}{\text{gain}|\text{select }b\in A}\Pr{}{\text{select }b\in A} +
        \E{}{\max_{i\not\in A}X_i}\\
        & \leq     \E{}{\text{gain}|\text{select }b\in A}\Pr{}{\text{select }b\in A}  + 1.5 \e \opt'\\
        & = \opt'' + 1.5 \e \opt',
    \end{align*}
    where in the last inequality we used Claim~\ref{cl:small_mean_vanishing_gain}. 
    \end{itemize}
    Combining these two steps, we get that $\opt'' \geq (1-1.5\e)^2 \opt$. We move on to bound the loss incurred by the two different discretization cases.
\begin{itemize}
    \item \textbf{Case 1}: Lemmas~\ref{lem:value_error_propagation} and \ref{lem:no_propagation_probability_error} hold with $\gamma=(1+\e)$, therefore we only lose a factor of $(1+\e)^2$, by applying first Lemma~\ref{lem:value_error_propagation} for the discretization in the values, and then Lemma~\ref{lem:no_propagation_probability_error} for the discretization in the probabilities.
    
    \item \textbf{Case 2}: Combining Claim~\ref{cl:single-large-value} and Lemma~\ref{lem:no_propagation_probability_error} for $\gamma=(1+\e)$, we transform all variables with support value $>1/\eps^2$ and discretize the probabilities to powers of $(1+\eps)$, losing $O(\eps)$ fraction of optimal reward.
\end{itemize}

Combining all the above together, we get $O(\eps)$ multiplicative
%$(1-1.5\e)^2(1-\e)^2(1+\e)^2$ \atd{remember to fix this after the change in Case 2} 
loss by the preprocessing and the discretization. 

\medskip\noindent\emph{Counting the groups:}
Observe that from the preprocessing phase all variables will belong to one group, i.e. the group where $v_i=0$ w.p. $1$.

\begin{itemize}
    \item \textbf{Case 1}: using the discretization described in Definition~\ref{def:discretization}, we have $O\lp( \frac{\log n/\e}{\log (1+\e)}\rp)$ different values for $p_i$ and $O\lp( \frac{\log 1/\e}{\log (1+\e)}\rp)$ different values for $v_i$. Therefore creating one group for each pair of values we have $O\lp( \frac{\log 1/\e \cdot \log n/\e}{\log^2 (1+\e)}\rp)$ different groups in total.
    
    \item \textbf{Case 2}: since all variables have the same value, we need to count the different groups from the discretization of the probabilities in range $\left[\frac{\e}{n}\tp, \tp \right]$. Recall that we discretize all probabilities by rounding them down to the nearest $(1+\e)^k$. The number of different $k$ values is $O\lp(\frac{\log n + \log (1/\e)}{\log (1+\e)}\rp)$.
\end{itemize}

Therefore the total number of groups  is $O\lp(\frac{\log^2 n/\e}{\log^2 (1+\e)}\rp)$.

This concludes the proof of the case of two-point variables. We describe next the extra step required to generalize the above arguments to the case of constant support size. We defer the discussion on the minor adjustments in the previous lemmas and claims that give the slightly different error bound for the constant size support (it now depends also on $c$) to Appendix \ref{apn:qptas_2_points}.

\paragraph{Counting the groups in the case of constant size support}
Observe that now \emph{groups} are defined differently than in the two-point case. In particular, to say that two random variables are in the same group in the DP, we need that they have the same distribution after discretization. This implies (1) that they end up with the same number of points after discretization, and (2) that the points fall in the exact same groups.

%\begin{proof}[Proof of Theorem~\ref{thm:QPTAS_const}]
%Similarly to Theorem~\ref{thm:QPTAS_2_point} we separately bound the error and calculate the number of groups resulting from our discretization process. 

   % For random variables $X_1, X_2, \dots, X_n$ with support size at most some constant $c$, we get 
   %  $O\lp( \lp(\frac{n}{g}\rp)^g \rp)$ for 
   %  $g =O\lp(\frac{\log^{2c} n/\e}{c! \log^{2c} (1+\e)}\rp)$ groups by following the discretization described above.

Suppose that after discretizing the values can take $k_1$ and the probabilities $k_2$ different values. Let us denote $k = k_1 \cdot k_2$ the different value-probability pairs that arise and by $\dist'_i$ the distribution of $X_i$ after the discretization. Assume for now that $|\text{supp}(\dist'_i)| =c$ for every $X_i$ and that also $|\text{supp}(\dist'_i)| =c$ (i.e., no two points of the support collapse to the same one). To count the number of different distributions that arise for support size $c$ we need a balls-into-bins argument, where the bins are the value-probability pairs and the balls are the points in the support. Then ${k \choose c}$ is the number of different distributions and two variables $X_i, X_j$ must have $\dist'_i = \dist'_j$ to be in the same group.
 
We now have to count all the different possibilities. Note that after discretization for every $X_i$ we have that $|\text{supp}(\dist'_i)| \in [1,c]$. For each of the possible support sizes, we count the emerging groups and add them to find the total number of groups. We, thus, obtain 
    \[
    \sum_{i=1}^c \binom{k}{i}\leq c \cdot \binom{k}{c} \leq \frac{k^c}{(c-1)!},
    \]
where the first inequality holds because the binomial is increasing in the interval $[1,c]$ for some constant $c$ (assuming $c = o(k)$), and the last inequality holds for every binomial coefficient by using Stirling's approximation. Since we follow the same discretization process as for two-point distributions, we know from Theorem~\ref{thm:QPTAS_2_point} that $k =O\lp(\frac{\log^2 n/\e}{\log^2 (1+\e)}\rp)$, and the desired result follows. Putting it all together, the number of groups is $g =O\lp(\frac{\log^{2c} n/\e}{(c-1)! \log^{2c} (1+\e)}\rp)$.
%\end{proof}

\subsection{QPTAS for general distributions}\label{subsec:qptas_general}
Moving on to the QPTAS for the general case, observe that the techniques used for the constant support case cannot work when the support size of a distribution is arbitrary. For instance, the transformation in the case of dealing with high values from before will not work now, because Claim~\ref{cl:single-large-value} cannot be applied when a variable can have multiple large support values $>1/\eps^2$. However, we are still able to obtain a QPTAS (Theorem~\ref{thm:QPTAS_general}): using the fact that whenever any (optimal) algorithm encounters a value $X_i >1$ it will accept it, we are able to reduce the size of the support significantly, using a new \emph{bundling} technique.

The second step involves the rest of the points with values in $[\e,1]$. These can possibly be arbitrarily many, but their range is constant, so we can discretize to powers of $(1+\e)$ as before. 
After doing that, we know that each $X_i$ has at most one point with $v_i>1$ (which can also be possibly unbounded) and a constant number of points in $[\e,1]$. Thus, we can employ the discretization process for the distributions with constant-size support and the main difference will now be that the support size $c$ depends on the fixed $\e$.

\begin{theorem}\label{thm:QPTAS_general}
There exists a $(1-\e)$-approximation algorithm for the prophet secretary problem against the online optimal that runs in time $O\lp( \lp(\frac{n}{g}\rp)^g \rp)$ for $g =O\lp(\frac{\log^{2c} n/\e}{(c-1)! \log^{2c} (1+\e)}\rp)$ where $c=O\lp( \frac{\log 1/\e}{\log (1+\e)}\rp)$.
\end{theorem}

\paragraph{Bundling \& Discretization:} recall that the normalized optimal satisfies $\opt \leq 1$ (from Claim~\ref{cl:bounded_opt_off}). For each variable $X_i$ we create a transformed variable $X'_i$ that is exactly the same as $X_i$ for values less than $1$, but we collapse the mass above $1$ to a single point with value the mean value of $X_i$ above $1$. Formally
\begin{equation}\label{eq:qptas_general_bundling}
X'_i = \begin{cases}
    x & \text{when } x \leq 1, \text{ w.p. }\Pr{X_i\sim \dist_i}{X_i=x}\\
    \E{X_i\sim \dist_i}{X_i | X_i>1} & \text{w.p. }\Pr{X_i\sim \dist_i}{X_i>1}.
\end{cases}
\end{equation}

The following claim formalizes the key observation that this bundling uses in order to avoid losing any gain, and Lemma~\ref{lem:bundling} shows that this bundling does not change the value of the optimal DP.

\begin{fact}\label{cl:algo_thresholds}
    The optimal DP for prophet secretary $\opt$, will not set a threshold more than $\opt$ at any step.
\end{fact}
To see why this holds, observe that at any point the threshold is the expected gain in the subproblem with one less variable. If this threshold was more than $\opt$, it would imply that the subproblem obtained a value higher than $\opt$.

\begin{restatable}{lemma}{qptasBundling}\label{lem:bundling}
Using the bundling described in \eqref{eq:qptas_general_bundling}, the value of the optimal DP solution does not change.
\end{restatable}
We defer the proof of the Lemma to Section~\ref{apn:qptas_general} of the Appendix.
\noindent
After transforming the variables as described above, for each variable, we use the following discretization process.

\begin{itemize}
    \item \textbf{Step 1}: We make $0$ all the values of the support that are $v_i<\e$.
    \item  \textbf{Step 2}: We discretize to powers of $(1+\e)$ all values in $v_i\in [\e, 1/\eps^2]$.
\end{itemize}

\begin{proof}[Proof of Theorem~\ref{thm:QPTAS_general}]
This case is reduced to the constant support size case of Theorem~\ref{thm:QPTAS_2_point} after the above steps.

\paragraph{Bounding the error:}
    Observe initially that from Lemma~\ref{lem:bundling} there is no loss incurred from the bundling process. 
    Using the same argument as Case 2 in the 2-point case, we can discard all values at most $\e$. Through this process, we only lose $O(\e)\opt$. For step 2, using Lemma~\ref{lem:value_error_propagation} we again incur loss of $O(\e)$.
    
    \paragraph{Counting the support:} From our discretization process, Step 2 creates $\frac{\log 1/\e}{\log (1+\e)}$ different values.

    Therefore each variable can have a support of size at most $\frac{\log 1/\e}{\log (1+\e)} +2$, where the extra $2$ is from the point of mass above $1$ and Case 1 (if the initial variable did not have $0$ in the support). 
    Now using Theorem~\ref{thm:QPTAS_2_point} with $c=O\lp( \frac{\log 1/\e}{\log (1+\e)}\rp)$ we get the theorem.
\end{proof}

Observe at this point that the techniques used also hold for continuous distributions. In particular, steps 1 and 2 from above are well-defined for continuous distributions. By applying these preprocessing steps, we have transformed the continuous distribution to a discrete one with constant support size. Moreover, it is easy to check that the same arguments as in the case of discrete distributions apply for bounding the error. Then, we can continue with the steps in the constant-support case (see \cref{subsec:qptas_2_point}) and obtain the same approximation guarantee. Since in the PTAS we start from the discretization done in the QPTAS, it is immediate that the PTAS results and techniques also hold for continuous distributions.

\section{PTAS}\label{sec:ptas}
In this section, we propose a PTAS for calculating an approximation of the optimal reward of the prophet secretary problem. In the QPTAS algorithm, we discretized the probabilities (and values) of the supports of each random variable to powers of $(1+\eps)$. However, even after the ``bundling'' (via rounding) process, for each random variable the number of possible realization probabilities of each support value can be $\Omega(\log n)$, and the support size of each variable can be a constant. Thus, the total number of variable groups can be $\polylog(n)$ and there can be more than $\poly(n)$ states in the dynamic program, which means that a polynomial time approximation scheme is not possible with only the discretization in previous sections. %\rojin{Fix this sentence. The way we are dealing with this we are considering q<eps20 but this contains poly(1/n). This sentence is present in the intro as well.}
In particular, we need to deal with the variables with realization probabilities that are small (say $\poly(\eps)$), but that there could be enough of them to make a non-negligible contribution to the optimal reward. This is problematic even in a simple setting where each variable is drawn from a two point distribution. 

To solve this, we observe that for a sequence of random variables with small realization probability, if their total realization probability is small ($O(\e^2)$), we can make a decision for the sequence as a whole: \textbf{even if we have seen the realization of all such random variables, we do not gain much compared to the original game where we need to make a decision at the arrival of each variable.} This technique, called \emph{frontloading} is the key for our proof, and allows us to obtain a PTAS for the general Prophet Secretary problem.

\begin{restatable}{theorem}{ptasMain}[PTAS for constant size support]\label{thm:PTAS_main} 
There exists a $(1-\e)$-approximation algorithm for the prophet secretary problem against the online optimal that runs in time $O(n^{(1/\eps)^{\poly(1/\eps)}})$. In the special case of 2-point distributions with $v_1 = 0$ and $v_2 > 0$, we can get running time of
$O(n^{poly(\frac{1}{\eps})})$. 
\end{restatable}

In the following sections, we formally show our main theorem. There is a summary of our notations in Table~\ref{table:cheatsheet} in the Appendix for convenience.

\subsection{A PTAS for binary distributions}\label{subsec: ptas binary small value}
We start with a simpler but representative case, where each variable's distribution $\dist_i$ is again defined as
\[X_i = 
\begin{cases}
v_i & \text{, w.p. } p_i\\
0 & \text{, w.p. } 1-p_i.
\end{cases}
\]
Using QPTAS transformations, we can assume:
 
\textit{Normalization. } Every $X_i$ either has $v_i\leq1/\eps^2$, or $v_i=v_{\max}$, where $v_{\max}$ is the largest possible support value of the original instance. Additionally, we may assume $\opt\in[0.669,1]$ ( Claim~\ref{cl:bounded_opt_off}).

\textit{Discretization. }
 There are only $\tilde{O}(1/\eps)$ different values of $v_i$ being either $v_{\max}$ or powers of $(1+\eps)$ between $\eps$ and $1/\eps^2$, losing at most $\eps$ fraction of optimal reward (Lemma \ref{lem:value_error_propagation}). We can also assume that every $p_i$ is a power of $(1+\eps)$ (Lemma \ref{lem:no_propagation_probability_error}). 

For every $v\in\mathbb{R}$, let $\Xs_{v}$ denote the set of random variables with the two support values being $0$ and $v$. Some of the proofs in this section are deferred to Section~\ref{apn:ptas_binary} of the Appendix.
% is drawn from a binary distribution $D_i$: with probability $p_i$, $X_i=v_i$ and with probability $1-p_i$, $X_i=0$. 
% Furthermore, we assume that each variable either has a small support value, that $v_i\leq1/\eps^2$, or support value $v_i=v_{max}$ being the largest possible support value of the original instance. 
%Furthermore, since we build on the QPTAS transformations, we can assume that every $X_i$ either has $v_i\leq1/\eps^2$, or $v_i=v_{\max}$, where $v_{\max}$ is the largest possible support value of the original instance. 
%Similarly to the QPTAS proof, we assume that the variables are normalized such that $\opt\in[0.669,1]$ (as in Claim~\ref{cl:bounded_opt_off}).
% the optimal reward is at most $1$ (see Claim~\ref{cl:bounded_opt_off}). 

\bigskip\noindent
\textbf{Small probability variables and special variables.}
 We group random variables that have ``small realization probabilities''. Let $q_i=-\ln(1-p_i)$. 
For any set of random variables $\Xs'_v=\{X_{i_1},X_{i_2},\cdots,X_{i_k}\}\subseteq \Xs_v$, the probability that $v$ gets realized from one of the random variables is exactly
\[1-\prod_{i_j}(1-p_{i_j})=1-\exp\left(\sum_{i_j}\ln(1-p_{i_j})\right)=1-\exp\left(-\sum_{i_j}q_{i_j}\right).\]
Although $q_i$ is not the true probability of realization, when $p_i$ is small enough, $q_i$ is a good approximation of $p_i$, and for any set $\Xs'_v$, $\sum_i q_i$ is a good approximation of the probability that $v$ gets realized in one of the variables in $\Xs'_v$. This is because for any $p_i<\eps$ we have that $p_i<q_i<p_i(1+\eps).$ This holds by writing down the Taylor series of $q_i$ based on $p_i$.

\paragraph{First Preprocessing.}
We initially split the variables into different groups, shown in Figure~\ref{fig:preprocessing_PTAS}. For each different support value $v$, the variables in $\mathbf{X}_{v}^*$ are ignored if their total probability is at most $\e^{10}$. Variables in $\mathbf{X}_v^+$ are discretized and treated like separate variables.

Specifically, for each support value $v\leq 1/\eps^2$, define a set of ``small probability variables'' $\Xs^*_v$ as
\[\Xs^*_v=\{X_i | X_i\in \Xs_v,\ s.t.\ q_i<\eps^{20}\}.\]
In other words, $\Xssv$ contains all random variables $X_i$ with support values $0$ and $v$, such that $q_i<\eps^{20}$.  
Since the probabilities of realization have been discretized to powers of $(1+\eps)$, the number of distinct $p_i$ with $-\ln(1-p_i)=q_i\geq \eps^{20}$ is only a constant (actually $\tilde{O}(1/\eps)$). 
Define $\Xslv=\Xsv\setminus\Xssv$, then it contains only $\tilde{O}(1/\eps)$ distinct random variables. We call $\Xslv$ ``\textbf{special variables} with support $v$'', and $\Xsl=\bigcup_{v}\Xslv$ the set of all special variables. There are only a constant number of distinct random variables in $\Xsl$.
Let $\Qsvo=\sum_{X_{i}\in \Xssv}q_i$ denote the sum of $q_i$ for all small-probability random variables with support values $0$ and $v$. Finally, let $\Qmax=\sum_{i: X_{i}\in X_{v_{\max}}}q_i$ and 
$\Xs^*_{v_{\max}}=\{X_i|X_i\in \Xs_{v_{max}},\ s.t.\ q_i<\eps^{20}\Qmax\}$.

\begin{restatable}{lemma}{qvLb}\label{lem:qv-lb}
For $v\leq1/\eps^2$, if $\Qsvo<\eps^{10}$ and for $v_{max}\geq \eps^{-2}$, if $\overline{Q}^*_{v_{max}}<\eps^{10}Q_{max}$, then ignoring the variables in $\Xssv$ reduces the optimal reward by at most $\eps^8$.
\end{restatable}

Using the above lemma, we can replace all variables in $\Xssv$ to zero values, where $\Qsvo<\eps^{10}$. Since there are only $\tilde{O}(1/\eps)$ distinct support values, the total loss in reward is negligible. From now on we assume $\Qsvo\geq\eps^{10}$ for every $v\leq 1/\eps^2$ and that $\overline{Q}^*_{v_{max}} \geq \eps^{10}Q_{max}$.

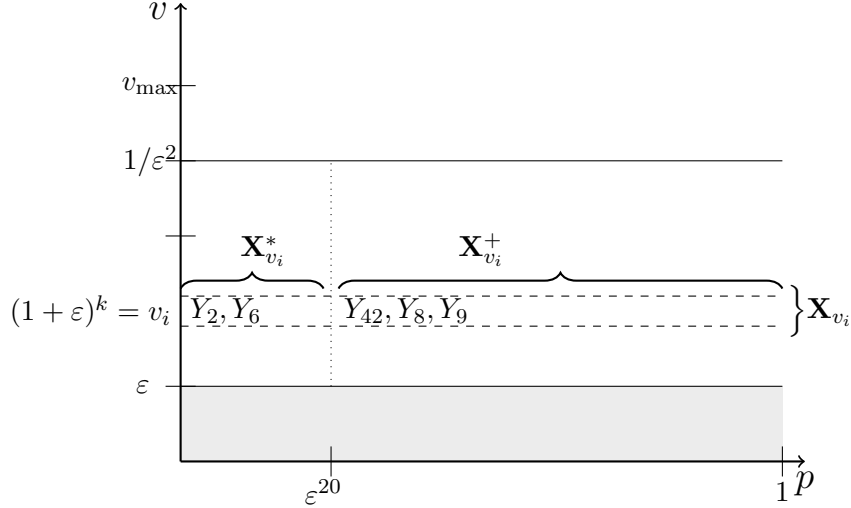
\begin{figure}[H]
    \centering
\begin{tikzpicture}

\pgfmathtruncatemacro{\casesL}{2};
\pgfmathtruncatemacro{\totalL}{8};
\pgfmathtruncatemacro{\vi}{2.5};

\tikzset{mycolor/.style={fill=gray!30, opacity=.5}}
%Fill shapes with color
\fill[mycolor] (0,0)--(0,1)--(\totalL,1)--(\totalL,0);

%Draw horizontal lines
\draw[->, thick] (0,0)--(0,\totalL-1.9) node[]{};
\draw[-] (0,1)--(\totalL,1) node[]{};
\draw[-] (0,4)--(\totalL,4) node[]{};

%Draw vertical lines
\draw[-, dotted] (\casesL,1)--(\casesL,4) node[]{};
\draw[->, thick] (0,0)--(\totalL+0.3,0) node[]{};

%y axis labels
\node[] at (-0.4,5) {$v_{\max}$};
\draw[-] (-0.2, 5) --(0.2,5);
\node[] at (-0.4,4) {$1/\e^2$};
\draw[-] (-0.2, 4) --(0.2,4);
\draw[-] (-0.2, 3) -- (0.2, 3);
\node[] at (-1.2,\vi) {$(1+\e)^k=v_i$};
\node[] at (-0.5,1) {$\e$};
\draw[-] (-0.2, 1) --(0.2,1);

%Draw box with variables
\draw[-, dashed] (0,\vi-0.2)--(\totalL,\vi-0.2) node[]{};
\draw[-, dashed] (0,\vi+0.2)--(\totalL,\vi+0.2) node[]{};
\node[] at (0.6,\vi) {$Y_{2}, Y_6$};
\node[] at (\casesL+1,\vi) {$Y_{42}, Y_8, Y_9$};
\node[] at (\totalL+0.5,\vi) {${ \Big \} } \mathbf{X}_{v_i}$};

%Legend on axes
\node[] at (-0.3,\totalL-2) {{\Large $v$}};
\node[] at (\totalL+0.3,-0.3) {{\Large $p$}};
	
%x axis labels
\node[] at (\casesL-0.1,-0.4) {$\e^{20}$};
\draw[-] (\casesL, -0.2) -- (\casesL, 0.2);
\draw[-] (\totalL, -0.2) -- (\totalL, 0.2);
\node[] at (\totalL, -0.4) {$1$};

%Draw braces + labels
\draw [decorate,black,thick,decoration={brace,amplitude=6pt}] (0,\vi+0.3) -- (\casesL-0.1,\vi+0.3) node[label={[xshift=-0.8cm, yshift=0.01cm]$\mathbf{X}^*_{v_i}$}]{};

\draw [decorate,black,thick,decoration={brace,amplitude=6pt}] (\casesL+0.1,\vi+0.3) -- (\totalL,\vi+0.3) node[label={[xshift=-4cm, yshift=0.01cm]$\mathbf{X}^+_{v_i}$}]{};

\end{tikzpicture}
\caption{First preprocessing of PTAS. In this example, variables $Y_2, Y_6, Y_{42}, Y_8$ and $Y_9$ have support value $v_i = (1+\e)^k$, while $Y_2$ and $Y_6$ have probability of being non-zero at most $\e^{20}$. If the total probability of variables in $\mathbf{X}_{v_i}^*$ is at most $\e^{10}$, the variables are ignored for the rest of the game.}
\label{fig:preprocessing_PTAS}
\end{figure}

\paragraph{Second Preprocessing.}
Of the variables in $\mathbf{X}_v^*$ that remained we will split them in blocks according to their probability as shown in Figure~\ref{fig:blocks_PTAS}. 
We begin by introducing some notation; we define $\Qsv=\min(\eps^{-1},\Qsvo)$ and $\Qov=\eps^4\Qsv$, and finally $\deltav=\frac{\Qov}{\Qsvo}$. 

% Before defining the other variants of the games, we define some useful terms. 
Let $Y_1,Y_2,\ldots,Y_n$ be the random realization of all variables, with $Y_i$ being the $i$'th random variable that arrives. 
For each support value $v>0$, partition \textbf{all} random variables 
to $\eps^{-4}+1$ blocks $\Bv_0, \Bv_{1}, \ldots, \Bv_{\eps^{-4}}$ such that each block contains a sequence of consecutive variables.
In particular, $\Bv_0$ contains the first $n\lp(1-\frac{Q^*_v}{\Qsvo}\rp)=n-\deltav n/\eps^4$ arrived variables (when $Q^*_v=\Qsvo$ we have $\Bv_0 =\emptyset$); each of $\Bv_{1},\ldots,\Bv_{\eps^{-4}}$ contains $\deltav n$ consecutive variables. 

\begin{figure}[H]
    \centering
    \begin{tikzpicture}
\pgfmathtruncatemacro{\totalL}{8};
\pgfmathtruncatemacro{\step}{1};
\pgfmathtruncatemacro{\blockZero}{3};

% \tikzset{mycolor/.style={fill=gray!30, opacity=.5}}
\pgfmathtruncatemacro{\height}{1};

%Draw rectangle
\draw [draw=black, thick] (0,0) rectangle (\totalL,\height);

%Draw blocks' limits
%block 0
%Fill shapes with color
\fill[pattern={Lines[angle=45, distance=6mm, line width=13mm, xshift=5mm]}, pattern color=gray!30,opacity=.5] ((0,0)--(\blockZero,0)--(\blockZero,\height)--(0,\height)--(0,0);

\draw[-] (\blockZero,0) -- (\blockZero,\height) ;

%Block 1
\draw[-] (\step+\blockZero,0) -- (\step+\blockZero,\height) ;

%Block 2
\draw[-] (\totalL-\step,0) -- (\totalL-\step,\height) ;

\node[] at (-1.7*\step, 0.5*\height) {Support value $\mathbf{v}$};

%Draw blocks nodes
\node[] at (0.5*\blockZero, 1.3*\height) {$\mathbf{B}_0^\mathbf{v}$};

\node[] at (0.5*\step+\blockZero, 1.3*\height) {$\mathbf{B}_1^\mathbf{v}$};

\node[] at (0.7*\totalL,0.5*\height) {$\dots$};

\node[] at (\totalL-0.5*\step, 1.3*\height) {$\mathbf{B}_{\e^{-4}}^\mathbf{v}$};

%Draw braces
\draw [decorate,black,thick,decoration={brace,amplitude=6pt, mirror}] (0,0) -- (\blockZero,0) node[label={[xshift=-1.2cm, yshift=-1cm]{\small $n\text{-}\delta_v n/\e^4$}}]{};

\draw [decorate,black,thick,decoration={brace,amplitude=6pt, mirror}] (\blockZero,0) -- (\blockZero+\step,0) node[label={[xshift=-0.4cm, yshift=-1cm]$\delta_v/n$}]{};

\draw [decorate,black,thick,decoration={brace,amplitude=6pt, mirror}] (\totalL-\step,0) -- (\totalL,0) node[label={[xshift=-0.3cm, yshift=-1cm]$\delta_v/n$}]{};

\end{tikzpicture}
    \caption{Second Preprocessing of the PTAS. For each support value $\mathbf{v}$ we split the variables into blocks $\Bv_i$. If $\Bv_0$ is not empty, we ignore it at game $G_3$.}
    \label{fig:blocks_PTAS}
\end{figure}
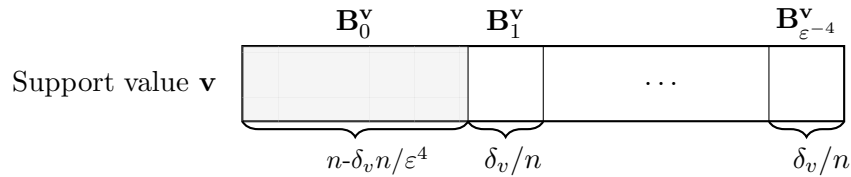

The intuition behind this partition is as follows. 
For all variables in $\Xssv$, the total $q_i$ of all variables in the last $\eps^{-4}$ blocks is roughly $Q^{*}_v$ with high probability, thus all variables in the first block $\Bv_{0}$ with small realization probability and support value 0 and $v$ can be ignored.
Furthermore, for each block $\Bv_{k}$, the total $q_i$ of variables in $\Xssv$ is roughly $\eps^{4}Q^*_v$, which is small. 
This means that even if all variables with small realization probability and support value $0$ or $v$ in each block are realized together, it does not influence the decision with high probability. 
This way, the small probability random variables are grouped, and we don't have to remember the identity of each individual random variable. For different $v$, the partition of the blocks can be different.\\

\textbf{High level proof sketch.}
After performing the two preprocessing steps. the main idea for our proof is to reduce our original prophet secretary problem to a problem with a smaller optimal DP solution, without losing too much. Our process is the following.

\begin{enumerate}
    \item \textbf{Game $G_1$.} We begin our reduction process, starting by game $G_1$ which corresponds to the original Prophet Secretary problem. The whole reduction process is shown in Figure~\ref{fig:summary_ptas}.
    \item \textbf{Game $G_2$.} In this game we remove the high variance variables. Specifically, we make sure that the realization prob of all variables in a block is \textbf{always} within ($1\pm \e)$ of their expectation. 
    
    \item \textbf{Game $G_3$.} We remove the small variables of block $\mathbf{B}_0^{\mathbf{v}}$ for each support value $v$. This will not decrease the reward by a factor more than $\epsilon$ of optimal, because by game $G_2$ we know that the probability of realization of a variable within the blocks except block zero is still considerable. This step is important, as it gives us an interval in which the realization probability of $v$ will be. %\atd{Write some intuition here? It works because we will still have a high chance of taking the $v$ later if we want to. We have to separate and probably $B_0$ has the significant part of mass of the total probability because each block cannot have too much total probability inside for Lemma 4.6 to work.}
    
    \item \textbf{Game $G_4$.} This is potentially the most crucial step, called \emph{frontloading}. In each block, and for each different support value $v$ we ``move" the uncertainty of the corresponding variables to the beginning of the block. Specifically, we flip a coin at the beginning of a new block, and if the flip succeeded we keep the value as an outside option.
    
    \item \textbf{Game $G_5$.} all outside options for $v$ have the same probability of realization. By doing this, we do not have to keep track of different probabilities for each outside option.
\end{enumerate}

\begin{figure}[H]
    \centering
    \pgfmathsetmacro{\dist}{2.5}
\pgfmathsetmacro{\distV}{2}
\begin{tikzpicture}

	\node (g1) at (0,0){{\LARGE $G_1$}};
 	\node (g2) at (\dist,0){{\LARGE $G_2$}};
  	\node (g3) at (2*\dist,0){{\LARGE $G_3$}};
   	\node (g4) at (3*\dist,0){{\LARGE $G_4$}};
    \node (g5) at (4*\dist,0){{\LARGE $G_5$}};

    %Draw all back arrows
    \draw[->,thick] (g2) to [in=315,out=225] node[below] {$(1-\e)\leq $}(g1);
   \draw[->,thick] (g3) to [in=315,out=225] node[below] {$(1-\e)\leq $}(g2);
    \draw[->,thick] (g4) to [in=315,out=225] node[below] {$\leq $}(g3);
   \draw[->,thick] (g5) to [in=315,out=225] node[below] {$(1-\e)\leq $}(g4);

   %Draw all forward arrows
	\draw[->,thick] (g1) to [in=135,out=45] node [above] {$\geq$} (g2);
    \draw[->,thick] (g2) to [in=135,out=45] node [above] {$\geq$} (g3);
    \draw[->,thick] (g3) to [in=135,out=45] node [above] {$\geq 1/(1+\e)$} (g4);
    \draw[->,thick] (g4) to [in=135,out=45] node [above] {$\geq$} (g5);

  %Draw theoremnodes
  \node[] (thmG12) at (0.5*\dist,0) {Lem~\ref{thm:equal-G1-G2}};

    \node[] (thmG23) at (1.5*\dist,0) {Lem~\ref{thm:equal-G2-G3}};

    \node[] (thmG34) at (2.5*\dist,0) {Lem~\ref{thm:equal-G3-G4}};

    \node[] (thmG45) at (3.5*\dist,0) {Lem~\ref{thm:equal-G4-G5}};
  
\end{tikzpicture}
    \caption{Series of reductions we use to prove Theorem~\ref{thm:PTAS_main}. $G_1$ is the original Prophet Secretary game, and $G_5$ is a new game that can be solved in polynomial time.}
    \label{fig:summary_ptas}
\end{figure}
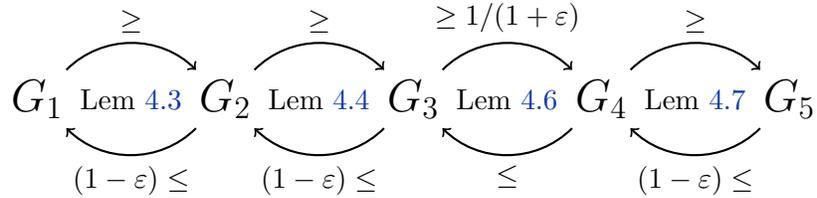

%%%%%%%%%%%%% G1 %%%%%%%%%%%%%
\bigskip
\noindent
\textbf{Original prophet secretary game $G_1$.}
There are $n$ random variables $X_1,X_2,\ldots,X_n$ arriving in random order. When each variable arrives, the gambler observes the realized value and decides whether to take it. If the gambler decides not to take the realized value, the value cannot be taken later. 

%%%%%%%%%%%%% G2 %%%%%%%%%%%%%
\bigskip
\noindent
\textbf{Prophet secretary game $G_2$ with possible failure.} The gambler plays prophet secretary game $G_1$, but the reward is set to $0$ if there exists a possible support value $v$ and a block $\Bv_k$ such that for the variables $X_{i_1},X_{i_2},\ldots,X_{i_{\ell}}\in \Bv_k\cap \Xssv$ we have that $\sum_{X_{i_j}\in \Bv_k}q_{i_j}\not\in[(1-\eps)\Qov,(1+\eps)\Qov]$. In other words, when the realization probability of $v$ of all variables with a small probability in a block is far from their expectation, no reward is given to the gambler.

\begin{lemma}\label{thm:equal-G1-G2}
$(1-\eps)\opt(G_1)\leq \opt(G_2)\leq\opt(G_1)$.
\end{lemma}

\begin{proof}
    First, we show the following claim.
    \begin{restatable}{claim}{smalErrSamples}\label{lem:smalErrSamples}
        Suppose we have $n$ non-negative real numbers $A_1,\cdots,A_n$ with $A_i\leq \eps^{10}\delta A$ for some $\delta<1$ where $A:=\sum_i A_i$. If we sample $m=\delta n$ numbers $Y_1,\cdots,Y_{m}$ uniformly at random and without replacement, then
        \[
        \Pr{}{|Y-\E{}{Y}|\geq \eps\E{}{Y} }<\eps^{8}.\]
    \end{restatable}

Now fix a block $\Bv_k$. In the lemma above, let $n$ be the total number of variables, $A=\Qsvo$ and $\delta=\delta_v$. Finally let $A_i=q_i$ for each variable $X_i\in \Xssv$. Note that the expectation of the sum of small probability variables is $\Qov$. Therefore, using the above lemma we conclude that the probability that the sum of $q_i$'s for small probability variables in a block (which corresponds to the $Y_i$'s) deviates from their expectation by more than a multiplicative factor of $\eps$ is at most $\eps^8$. Now there are $\eps^{-4}$ many blocks, and there are at most $\Tilde{O}(1/\eps)<\eps^{-2}$ many different support values. Using union bound, we can conclude that the event that prompts the gambler not to get any award in $G_2$ happens with probability at most $\eps^8\cdot \eps^{-4}\cdot \eps^{-2} = \eps^2$. This concludes the proof of the lemma. 
\end{proof}

%%%%%%%%% G3 %%%%%%%%%%%%%

\bigskip
\noindent
\textbf{Prophet secretary game $G_3$ with lazy decision.} The gambler plays game $G_2$. However, every variable with support $0$ and $v$ and small realization probability in $\Xssv\cap \Bv_0$ is not realized.

\begin{lemma}\label{thm:equal-G2-G3}
$(1-\eps)\opt(G_2)\leq \opt(G_3)\leq\opt(G_2)$.
\end{lemma}

Before moving onto the proof of the lemma, we state an  observation used in the proof. 

\begin{restatable}{observation}{skippingInitial}\label{skipping initial}
Let $\Qsv=\min(\eps^{-1},\Qsvo)$. If we encounter a variable $X_i \in \Xssv$ while the remaining variables in $\Xssv$ have total $q_i$ more than $\Qsv$, then ignoring $X_i$ decreases the reward of the algorithm by a factor of at most $\eps$. 
\end{restatable}

\begin{proof}[Proof of Lemma~\ref{thm:equal-G2-G3}]
By the definition of $G_2$, we know that the effective sum\footnote{Here by effective sum we mean that the cases explained in $G_2$ are ignored.} of $q$ values of $X_i$ variables in $\Bv_k\cap \Xssv$ for any $k$ is in the interval $[(1-\eps)\Qov, (1+\eps)\Qov]$. The total number of blocks, other than the first one, is $\eps^{-4}$. This means the total sum of $q$ values of these variables is in the interval $\eps^{-4}[(1-\eps)\Qov, (1+\eps)\Qov] = Q^*_v[(1-\eps), (1+\eps)]$. Now using Observation \ref{skipping initial}, since we know that if we skip block $\Bv_0$, the sum of $q_i$s of the remaining variables is at least $(1-\eps)Q^*_v$, the loss in reward can be at most $\eps(1-\eps)=O(\eps^2)$. This concludes the proof. 
\end{proof}

%%%%%%%%% G4 %%%%%%%%%%%%%
\bigskip
\noindent
\textbf{Prophet secretary game $G_4$ with lazy decision and outside option.} The gambler plays prophet secretary game $G_3$. However, in each block $\Bv_k$ with $k\geq 1$, the realization of variables with small probability in $\Xssv\cap \Bv_k$ is different. In particular, all variables in $\Xssv\cap \Bv_k$ are realized at the beginning of $\Bv_k$ (and not later), and kept as an outside option $Z_{v}$. Whenever the gambler decides to accept a variable $Y_i$ in $\Bv_k$, the gambler can choose to replace $Y_i$ by $Z_{v}$. At the end of $\Bv_k$, the gambler can take $Z_v$ and stop. If the gambler decides to take nothing and continue the game, the outside option $Z_v$ is cleared.

\begin{lemma}\label{thm:equal-G3-G4}
$\opt(G_3)\leq \opt(G_4)\leq(1+\eps)\opt(G_3)$.
\end{lemma}
\begin{proof}[Proof of Lemma~\ref{thm:equal-G3-G4}]

We prove this lemma by induction on the number of support values that have been grouped. For any set $\Vs$ of support values, let the game $G_3(\Vs)$ be $G_3$ in which only variables with support values $v\in \Vs$ are grouped. Then $G_4=G_3(\{\textrm{all support values}\})$. First, we show that $\opt(G_3)\leq \opt(G_3(\{v\}))\leq (1+O(\eps^3))\opt(G_3)$ (\textit{base case}). Notice that the exact same proof works for showing that for any $\Vs$, $\opt(G_3(\Vs))\leq \opt(G_3(\Vs\cup\{v'\}))\leq (1+O(\eps^3))\opt(G_3(\Vs))$ when an additional support value $v'$ is grouped (\textit{inductive step}). Then by induction $\opt(G_3)\leq \opt(G_4)\leq(1+\eps)\opt(G_3)$, as there are only $\tilde{O}(1/\eps)<1/\eps^2$ distinct discretized support values.

Denote $\Gv=G_3(\{v\})$. Notice that the gambler in game $G_3$ has strictly less information than the gambler in game $\Gv$: in any block $\Bv_k$ in $\Gv$, compared to $G_3$, the variables with small probability and support value $v$ are front-loaded and realized together at the beginning of the block, and exist as an outside option throughout the entire block. Thus the optimal reward of $\Gv$ is always at least the optimal reward of $G_3$, i.e. $\opt(G_3)\leq \opt(\Gv)$. In the remainder of the proof, we show that $\opt(\Gv)\leq (1+O(\eps^3))\opt(G_3)$. In particular, we show how to transform an optimal algorithm $\mathcal{A}_v$ of Game $\Gv$ to an algorithm $\mathcal{A}$ of Game $G_3$, while the reward is almost unchanged. As the optimal decision of $\Gv$ can be calculated via a dynamic program, we can assume without loss of generality that $\mathcal{A}_v$ is a threshold-based algorithm (using the optimal online reward of the remaining instance at each step), with the threshold being non-increasing over time. 

\textbf{Coupling. }We consider the decision trees of the two games. For each node in the decision tree of $G_3$, there is a corresponding node in the decision tree of $\Gv$, where the realization histories of all variables are identical. We completely couple the randomness of the two games, so that when we analyze the decision trees of the two games, the arrival order of the variables and the realized value of each variable are the same. 

\textbf{Building $\mathcal{A}$ from $\mathcal{A}_v$. }
\begin{enumerate}
    \item \textit{On block $\Bv_0$, algorithm $\mathcal{A}$ makes exactly the same decision in $G_3$ as the algorithm $\mathcal{A}_v$ in game $\Gv$. } The expected reward from $\Bv_0$ in the two games will be the same. This is because in going from game $G_2$ to $G_3$, we made the assumption that the reward from all the small probability variables of support $v$ in $\Bv_0$ are zeroed out.

\item On block $\Bv_k$, consider any node $N_v$ in the decision tree when the gambler reaches $\Bv_k$ and let $N$ be the corresponding node in the decision tree of $G_3$. Depending on how $\mathcal{A}_v$ performs, we define $\mathcal{A}$:

\end{enumerate}

\textbf{\textit{Case 1. }}In $\mathcal{A}_v$, after the gambler reaches decision node $N_v$, she always wants to accept $v$ whenever it is available in $\Bv_k$. Consider the following algorithm $\mathcal{A}$ for $G_3$:
\begin{itemize}
    \item Whenever the gambler observes $v$ realized by a variable in $\Xssv$, the algorithm accepts it;
    \item Whenever the gambler observes a realized value not in $\Xssv$, the gambler makes the same decision as in the corresponding node in $\Gv$, that if the $\mathcal{A}_v$ accepts the same value (possibly changed to an outside option $v$) in $\Gv$, $\mathcal{A}$ also accepts the variable in $G_3$; if $\mathcal{A}_v$ does not accept the same value in the corresponding node in $\Gv$, $\mathcal{A}$ also rejects the variable.
\end{itemize}

\textit{Analysis. }We analyze how much reward difference the two algorithms have in $\Bv_k$. 
\begin{enumerate}[(a)]
\item Since $\mathcal{A}_v$ always accepts outside option $v$ at the end of block $\Bv_k$, the probability that $\mathcal{A}_v$ accepts some variable in $\Bv_k$ is identical to the probability that $\mathcal{A}$ accepts some variable in $\Bv_k$.
\item Any time $\mathcal{A}_v$ accepts some variable $X_i$ before the end of $\Bv_k$ in $\Gv$, $\mathcal{A}$ either accepts $X_i$ in the corresponding node in $G_3$, or it accepts $v$ from a variable in $\Xssv$ before reaching the node (and at the same time $\mathcal{A}_v$ can accept $\max(v,X_i)$). This happens with probability at most $(1+\eps)\Qov<2\Qov<2\eps^{3}$, as the total $q_i$ of all variables in $\Xssv$ in $\Bv_k$ is at most $(1+\eps)\Qov$. 
\item Any time $\mathcal{A}_v$ accepts the outside option $v$ at the end of $\Bv_k$ in $\Gv$, $\mathcal{A}$ accepts $v$ in $G_3$.
\end{enumerate}

This way, we have coupled the reward of the two algorithms in the two games, such that whenever $\mathcal{A}_v$ accepts some value, $\mathcal{A}$ accepts a possibly different value 
with probability at most $2\eps^{3}$. Thus the reward of $\mathcal{A}$ in block $\Bv_k$ in Game $G_3$ is at least $(1-2\eps^3)$ fraction of the reward of $\mathcal{A}_v$ in block $\Bv_k$ in Game $\Gv$, starting from corresponding nodes $N$ and $N_v$ in the two decision trees. %\rojin{shouldnt reward of $A_v$ from small variables be more than that of $A$ since we did front loading?}
\\
\textbf{\textit{Case 2}}. If in $\mathcal{A}_v$, after the gambler reaches decision node $N_v$ the gambler \textbf{does not} always want to accept $v$ whenever it is available in $\Bv_k$. 
 Consider the following algorithm $\mathcal{A}$ for $G_3$:
 \begin{itemize}
     \item  Whenever the gambler observes $v$ realized by a variable in $\Xssv$, the algorithm \textbf{rejects} it;
     \item whenever the gambler observes a realized value not in $\Xssv$, the gambler makes the same decision as in the corresponding node in $\Gv$ in a way the same as Case 1.
 \end{itemize}
 
\textit{Analysis. The following cases may occur:}
\begin{enumerate}[(a)]
\item When a variable $X_i$ not in $\Xssv$ arrives in $\Gv$, and $\mathcal{A}_v$ decides to accept it, $\mathcal{A}$ also accepts $X_i$ in $G_3$ with the same value. Since with probability at most $(1+\eps)\Qov<2\Qov<2\eps^{3}$, the value accepted by $\mathcal{A}_v$ is replaced by $\max(v,G_v)$, we have coupled the two games such that the reward of $\mathcal{A}$ is at least $(1-2\eps^3)$ fraction of the reward of $\mathcal{A}_v$.
\item When $\mathcal{A}_v$ reaches the end of the block $\Bv_k$ at some decision node $N_{v,1}$, and $\mathcal{A}_v$ decides to accept the outside option whenever it is $v$, $\mathcal{A}$ does not accept any variable in $\Bv_k$. Therefore, the reward loss from $\mathcal{A}$ in $G_3$ is $v$ with probability at most $2\Qov$, conditioned on that $\mathcal{A}_v$ reaches decision node $N_{v,1}$. In the future, as $\mathcal{A}_v$ already decides to accept $v$ at $N_{v,1}$, it will always decide to accept $v$ whenever available in the decision tree rooted at $N_{v,1}$, since the optimal threshold based algorithm always sets non-increasing thresholds. Thus in the decision tree rooted at the corresponding node of $N_{v,1}$ in $G_3$, this case will not be discussed, which means that any decision node like $N_{v,1}$ discussed in this case cannot be the ancestor of another such node.
\item When $\mathcal{A}_v$ reaches the end of the block $\Bv_k$ at some decision node $N_{v,2}$, and $\mathcal{A}_v$ decides not to accept the outside option even if it is $v$, both algorithms do not accept any variable in $\Bv_k$, thus both have reward 0.
\end{enumerate}
Now we are ready to analyze the reward of $\mathcal{A}$ in $\Bv_k$ of $G_3$. It is at least $(1-2\eps^3)$ times the reward of $\mathcal{A}_v$ in $\Bv_k$ of $\Gv$, minus $\Pr{}{\mathcal{A}_v\textrm{ reaches a node like }N_{v,1}\textrm{ in Case 2(b)}}\cdot 2\Qov \cdot v$. Thus summing the reward for all of the blocks, we have
\[ALG(\mathcal{A},G_3)\geq (1-2\eps^3)\opt(\Gv)-\sum_{k}\sum_{N_{v,1}}\Pr{}{\mathcal{A}_v\textrm{ reaches }N_{v,1}\textrm{ in Case 2(b) in }\Bv_k}\cdot 2\Qov v.\]
As we discussed in Case 2(b), no decision node like $N_{k,1}$ is the ancestor of another decision node in Case 2(b) in later blocks. Thus $\sum_{k}\sum_{N_{v,1}}\Pr{}{\mathcal{A}_v\textrm{ reaches }N_{v,1}\textrm{ in Case 2(b) in }\Bv_k}\leq 1$, which implies
\begin{equation}\label{eqn:alg-g3}
ALG(\mathcal{A},G_3)\geq (1-2\eps^3)\opt(\Gv)-2\Qov v.
\end{equation}
Now we bound $\Qov v$. Observe that if an algorithm only accepts realized value $v$ from $\Xssv$, as the realization probability of $v$ is at least $(1-e^{-\Qsv})$, its reward is at least $(1-e^{-\Qsv})v$. Thus
$(1-e^{-\Qsv})v\leq \opt(G_3).$
Also notice that when $1\leq \Qsv\leq 1/\eps$, $1-e^{-\Qsv}\geq 1-e^{-1}$, while $\Qov=\eps^4\Qsv\leq \eps^3$, thus $\Qov v\leq \frac{e}{e-1}\eps^{3}\opt(G_3)$; when $\Qsv<1$, $1-e^{-\Qsv}>\frac{1}{2}\Qsv$, thus $\Qov v=\eps^4\Qsv v<2\eps^4\opt(G_3)$. Therefore, $\Qov v=O(\eps^3)\opt(G_3)$ always holds, so by \eqref{eqn:alg-g3} we have 
\[\opt(G_3)\geq ALG(\mathcal{A},G_3)\geq (1-O(\eps^3))\opt(\Gv),\]
which completes the proof.
\end{proof}

%%%%%%%%%%% G5 %%%%%%%%%%%%%
\bigskip
\noindent

\textbf{Prophet secretary game $G_5$ with perturbed outside option.} The gambler plays the prophet secretary game $G_4$. However, in each block $\Bv_k$ with $k\geq 1$, the outside option $Z_v$ is set to $v$ with a fixed probability $1-e^{-(1-\eps)\Qov}$.

\begin{lemma}\label{thm:equal-G4-G5}
$(1-\eps)\opt(G_4)\leq \opt(G_5)\leq\opt(G_4)$.
\end{lemma}

The proof of this lemma can be found in \ref{thm:equal-G4-G5-general}, and is identical to the proof of the similar lemma for the general distribution case.
%\evagerg{we don't have a proof at all for this lemma (previously theorem). if it's the exact same as the D5 for the general distribution case (it probably is?) we should say smth at least}

Combining Lemmas~\ref{thm:equal-G1-G2}, \ref{thm:equal-G2-G3}, \ref{thm:equal-G3-G4} and \ref{thm:equal-G4-G5} above, we get the following corollary, showing that Game $G_5$ obtains almost the same optimal reward as the original prophet secretary problem.

\begin{corollary}\label{thm:equal-G1-G5}
    $(1-O(\eps))\opt(G_1)\leq \opt(G_5)\leq(1+O(\eps))\opt(G_1)$.
\end{corollary}

\begin{theorem}\label{DP}
    Prophet secretary game $G_5$ can be optimally solved in time $n^{\poly(\frac{1}{\eps})}$.
\end{theorem}
\begin{proof}[Proof of Theorem~\ref{DP}]
    We propose a dynamic program that solves $G_5$ in $n^{\poly(\frac{1}{\eps})}$ time. 

    First, recall our assumption that for any value $v$ in support of some variable, we have $\eps \leq v\leq \eps^{-2}$.
    Moreover, we discretized the space, so that there are only $\Tilde{O}(\frac{1}{\eps})$ different values of $v$, and that all probabilities $p_i$ are powers of $(1+\eps)$ and belong in $[\eps^{20}, 1]$ for $v_i\leq 1/\eps^2$, and $p_i\in[\eps^{20}\Qmax/(1+\eps),\Qmax]$ where $\Qmax=\sum_{i: X_{i}\in X_{v_{\max}}}q_i$. 
    At any point in the sequence, let $n_{v,p}$ be the number of remaining variables with support value $v$ and probability of realization of $p$.
    We claim that the information the optimal algorithm needs to make a decision at each point is the value of $n_{v,p}$ for all $v$ and $p$, and the values of $Z_v$, the outside option of the current block of $v$, for every value $v$. 
    This is because if the algorithm is facing variable $X$, and it is of the group $(v,p)$, then the choice it makes depends on the realized value of $X$ (whether it is $0$ or $v$), the outside option $Z_v$, and the expected optimal value of the future, given that the number of variables of group $(v,p)$ is one less than before.
    So if 
    $\mathcal{N}=\{n_{v,p}\}_{\forall v, \forall p}$ and $\mathcal{Z}=\{Z_v\}_{\forall v}$, then one can define
    \begin{equation*}
        \opt(\mathcal{N}, \mathcal{Z}) = \sum_{v,p} \Pr{}{\textrm{next variable is X=(v,p)}}\cdot\mathbb{E}_{X}\left[\max(X, \max_v Z_v, \opt(\mathcal{N}_{v,p},\mathcal{Z}'))\right]
    \end{equation*}
    where $\mathcal{N}_{v,p}$ is the same as $\mathcal{N}$, only that its corresponding value for $n_{v,p}$ is deducted by 1. Moreover, $\mathcal{Z}'$ is the same as $\mathcal{Z}$, unless we are at a ''critical point'', where we get inside the next block $\Bv_k$ for some support value $v$. In order to obtain $\mathcal{Z}'$ from $\mathcal{Z}$ we do the following; (1) if there's no new group starting it remains the same (2) if there is a new group starting, we change the outside option of the group that ended to $0$, and flip again the coin to see if we still have the variable $v$ available. 
    
    \bigskip
    \noindent
    \textbf{Size of the DP: } The total number of values in support is at most $\frac{\log 1/\e}{\log (1+\e)}=\tilde{O}{(1/\eps)}$ and the total number of possible probabilities is at most $20\frac{\log 1/\e}{\log (1+\e)}=\tilde{O}{(1/\eps)}$. Moreover, for all $v$ and $p$, $0\leq n_{v,p}\leq n$, and the value of $Z_v$ is either $0$ or $v$. This means the first component $\mathcal{N}$ of DP has at most $n^{\poly(\frac{1}{\eps})}$ possible choices, and the second component $\mathcal{Z}$ has $2^{\tilde{O}(1/\eps)}$ possibilities. If at the beginning of the sequence, the values are stored in $\mathcal{N}_0$ and $\mathcal{Z}_0$, then finding $\opt(\mathcal{N}_0,\mathcal{Z}_0)$ gives us the optimal answer. This takes time at most $O(n^{\text{poly}(1/\e)})$, yielding a PTAS. 
    % \evagerg{do we need to say $O(n^{\text{poly}(1/\e)})$ here? since we multiply by $2^{...}$}
\end{proof}

\begin{corollary}
There exists a PTAS algorithm for the prophet secretary game $G_1$. This is a direct result of  \Cref{thm:equal-G1-G5} and \Cref{DP}.
\end{corollary}

\begin{remark}
The main idea in generalizing this scheme to general distributions is that we transform each distribution with support of size $k>2$, to $k$ distributions in the 2 point form, where each is one of the support values with its corresponding probability and the rest of the mass is on zero. These distributions arrive sequentially. 
\end{remark}

\bibliography{refs}
\bibliographystyle{plainnat}

\appendix 
 \section{Why Simple Grouping Ideas Fail}\label{apn:counterexample}
Our initial idea for grouping ``similar" variables was to group the ones that have similar (within $\e$) expected value and variance, to obtain the groups shown in the grid shown in Figure~\ref{fig:grouping_counterexample}. Ideally grouping the variables like that, and treating variables that ``fall" in the same square of the grid as the same will not cost more than $O(\e)$ in the optimal DP.

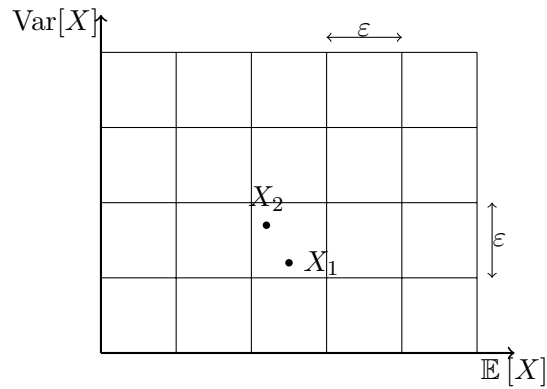
\begin{figure}[H]
    \centering
    \begin{tikzpicture}

\draw[step=1.0,black,thin] (0,0) grid (5,4);

\draw[->, thick] (0,0) -- (0,4.5);

\draw[->, thick] (0,0) -- (5.5,0);

\node[] at (5.5,-0.25) {$\E{}{X}$};

\node[] at (-0.6, 4.4) {$\text{Var}[X]$};

\draw (2.5, 1.2) node[circle,fill=black,inner sep=1pt,label=right:$X_1$]{};
\draw (2.2, 1.7) node[circle,fill=black,inner sep=1pt,label=above:$X_2$]{};

\draw[<->] (5.2,1) -- (5.2,2);
\node[] at (5.3, 1.5){$\e$};

\draw[<->] (3,4.2) -- (4,4.2);
\node[] at (3.5, 4.3){$\e$};

\end{tikzpicture}
    \caption{Simple grouping according to mean and variance that does not work.}
    \label{fig:grouping_counterexample}
\end{figure}

However consider the following example; there are two variables $X_A$ and $X_B$ defined as follows. 
\[ X_A = \begin{cases}
3/4, & \text{w.p. } 1/2\\
1/4, & \text{w.p. } 1/2
\end{cases}, \hspace{0.2cm}
X_B = \begin{cases}
    1, & \text{w.p. } 1/8\\
    1/2, &\text{w.p. } 3/4\\
    0, &\text{w.p. } 1/8
\end{cases}
\]

and observe that $\mu = \E{}{X_A} = \E{}{X_B} = 1/2$ and $\text{Var}(X_A) = \text{Var}(X_B) = 1/16$, therefore they would be treated as the same. 

In an instance with 2 variables only, if we calculate the DP solutions of having a $X_A$ variable first and having a $X_B$ variable first we get 
$\E{}{\opt(X_A, \mu)} = \E{}{\opt(X_A, 1/2)} = 5/8 = 0.625$ while $\E{}{\opt(X_B, \mu)} = \E{}{\opt(X_B, 1/2)} = 9/16 \approx 0.563$. These solutions differ by a constant even though the gap of the mean and variance of the variables is $0$, therefore we cannot hope for any qptas/ptas algorithm.

\section{Proofs from Section~\ref{sec:prelims}} \label{apn:prelims}

\errorProp*
\begin{proof}
We prove this using induction on the number of variables. When there is a single variable, then $\alg$ and $\opt$ both use $0$ as the threshold for $X_1$, thus having the same reward. Suppose that there are $k+1$ variables, and without loss of generality assume the first variable arrived is $X_1$. Notice that $\alg$ uses $\opt'(X_{>1})$ from the remaining variables $X_{>1}$ as the threshold for $X_1$. Then we can calculate the reward of $\alg$ as follows:
\begin{eqnarray*}
\alg(X_1,X_{>1}) &=& \alg(X_{>1})\Pr{}{X_1<\opt'(X_{>1})}\\
& &+ \mathbb{E}[X_1|\opt'(X_{>1})\leq X_1< \opt(X_{>1})]\cdot\Pr{}{ \opt'(X_{>1})\leq X_1< \opt(X_{>1})}\\
& &+ \mathbb{E}[X_1|X_1\geq \opt(X_{>1})]\cdot\Pr{}{X_1\geq \opt(X_{>1})}\\
&\geq& (1-\eps)\opt(X_{>1})\Pr{}{X_1<\opt'(X_{>1})}\\
& &+ \opt'(X_{>1})\cdot\Pr{}{ \opt'(X_{>1})\leq X_1< \opt(X_{>1})}\\
& &+ \mathbb{E}[X_1|X_1\geq \opt(X_{>1})]\cdot\Pr{}{X_1\geq \opt(X_{>1})}\\
&\geq&(1-\eps)\opt(X_{>1})\Pr{}{X_1<\opt(X_{>1})}\\
& &+ \mathbb{E}[X_1|X_1\geq \opt(X_{>1})]\cdot\Pr{}{X_1\geq \opt(X_{>1})}\\
&\geq&(1-\eps)\opt(X_1,X_{>1}).
\end{eqnarray*}
Here the first line is true by discussing the expected reward in different cases of $X_1$; the second line is true since by induction on $k$ remaining variables $X_{>1}$, $\alg(X_{>1})\geq (1-\eps)\opt(X_{>1})$; the third line is true since $\opt'(X_{>1})\geq (1-\eps)\opt(X_{>1})$. The last line is true since the optimal algorithm sets $\opt(X_{>1})$ as the threshold for $X_1$. Thus by induction $\alg(X)\geq (1-\eps)\opt(X)$ holds for any number of variables.

\end{proof}

\clGroups*
\begin{proof}[Proof of Claim~\ref{cl:DP_group_size}]

The way we defined the DP, the decision tree has one vertex for each combination of the number of remaining variables within each group. Let $S_i := [0, |G_i|]$, i.e., 
$S_i$ contains all possible values for the cardinality of the yet unobserved variables of a group during a run of the DP. 
Note that for the set of vertices (or the size of the DP) $V$ we have $|V| = \prod_{i=1}^{g} |S_i|$. Therefore, we want to split our $n$ variables into those $g$ groups. If at step $j$ we observe a variable $X_j \in \groups_i$, then the next state has $|\groups_i^{j+1}| := |\groups_i^{j}-1|$, where $\groups_i^{j}$ denotes the set of yet unobserved random variables belonging to the group before time $j$. When no variables remain in a group $\groups_i$, it is simply the empty set.

Now we show that $|V|$ is maximized when the variables are split equally among the groups, meaning that $|S_i| = \left\lceil \frac{n}{g} \right\rceil$ for all $i$. Assume that $|V|$ is maximized for some other split of the variables which will result in some other values $|\groups'_i|, i \in [n]$. Start again from the configuration where all groups have the same cardinality. Pick any pair $\groups_j, \groups_k$ and move $\ell$ variables from one group to the other, so that now $|\groups'_j|= \left\lceil \frac{n}{g} \right\rceil - \ell$ and $|\groups'_k|= \left\lceil \frac{n}{g} \right\rceil + \ell$, for some $\ell \in \mathbb{N}$. Then $|V| = \left\lceil \frac{n}{g} \right\rceil ^{g-2} \left( \left\lceil \frac{n}{g} \right\rceil - \ell \right) \left( \left\lceil \frac{n}{g} \right\rceil + \ell \right) < \left\lceil \frac{n}{g} \right\rceil ^{g}$, which is a contradiction.
\end{proof}

\errorPropValue*
\begin{proof}[Proof of Lemma~\ref{lem:value_error_propagation}]
We use induction on the number of variables that remain in the game, and show the lemma for the multiplicative case. The additive case follows exactly the same steps. 

For the base case, let $X$ be the variable that is currently instantiated, and $Y$ the single variable that remains. Recall that $\opt$ stops when $X \geq \E{}{Y} = \sum_i v_i \Pr{}{Y=v_i}$, and similarly $\opt'$ stops when $X \geq \E{}{Y'} = \sum_i v'_i \Pr{}{Y=v'_i}$. Since $v'_i \leq v_i \leq \gamma\cdot v'_i$ for every $v_i \in \supp(Y)$ it holds that
\begin{equation}\label{eq:base_case_means}
    \frac{\E{}{Y}}{\gamma} \leq \E{}{Y'} \leq \E{}{Y}.
\end{equation}
We distinguish in the following cases.

\begin{itemize}
    \item Cases 1 \& 2 : $\opt_1$ and $\opt_1'$ either both stop and take $X$ or both continue and take $Y$, in which case their values are the same.
    \item Case 3: $\opt_1'$ stops and takes $X$, while $\opt_1$ continues and gets $\E{}{Y}$. Using equation~\eqref{eq:base_case_means} we have that $\opt_1' = X \geq \E{}{Y'} \geq \frac{\E{}{Y}}{\gamma} = \opt_1/\gamma$.
    \item Case 4: $\opt'$ continues and receives $\E{}{Y}$ while $\opt_1$ stops and receives $X$. Then $\opt'_1 = \E{}{Y} \geq \E{}{Y'} \geq X = \opt_1$, where we used that since $\opt'_1$ continued it holds that $\E{}{Y'}\geq X$.
\end{itemize}

For the induction, we assume the claim holds for at most $n$ variables, and show it for $n+1$. Let $X_1$ be the variable that is currently instantiated. We distinguish in the following cases.

\begin{itemize}
    \item Cases 1 \& 2: $\opt_{1:n+1}$ and $\opt_{1:n+1}'$ either both stop and take $X_1$ or both continue and take $\opt_{2:n+1}$, in which case their values are the same.
    \item Case 3: $\opt_{1:n+1}'$ stops and takes $X_1$, while $\opt_{1:n+1}$ continues and gains $\opt_{2:n+1}$ on expectation, then  we get that
     \begin{align*}
     \opt_{1:n+1}' & = X_1 & \text{since }\opt' \text{ stopped}\\
     & \geq \opt'_{2:n+1}  & \opt' \text{ is a threshold policy}\\
     & = \frac{1}{n} \sum_{i\in [n]} \opt'(X_i| X_{-i}) & \text{Def. of }\opt' \eqref{eq:DP}\\
   & \geq \frac{1}{\gamma} \cdot \frac{1}{n} \sum_{i\in [n]} \opt(X_i| X_{-i}) & \text{Induction hypothesis}\\
     &= \frac{\opt_{2:n+1}}{\gamma} = \frac{\opt_{1:n+1}}{\gamma} & \text{Def~\eqref{eq:DP} and $\opt$ stops} 
     \end{align*}
    
    \item Case 4: $\opt'_{1:n+1}$ continues while $\opt_{1:n+1}$ stops and receives $X_1$. Since $\opt'_{1:n+1}$ continues we have that $\opt'_{1:n+1} = \frac{1}{n} \sum_{i\in [n]} \opt'(X_i| X_{-i}) \geq X_1 = \opt_{1:n+1}$. 
\end{itemize}
\end{proof}

\errorProb*
\begin{proof}[Proof of Lemma~\ref{lem:no_propagation_probability_error}]
Observe that in the proof of Lemma~
\ref{lem:value_error_propagation} holds using exactly the same steps for the main induction. The only step that uses the fact that we have different values is the base case, which holds the same when the guarantee holds for $p_i'$ instead of $v_i'$.
\end{proof}

\section{Proofs from Section~ \ref{sec:qptas}}\label{apn:qptas}
\subsection{Proofs from Section~\ref{subsec:qptas_2_point}}\label{apn:qptas_2_points}

\vanishingGain*
\begin{proof}[Proof of Claim~\ref{cl:small_mean_vanishing_gain}]
We will try to bound the expected contribution of these variables in the DP. The worst case will be if every $X_i$ has $\E{}{X_i} \leq \frac{\e}{n}$ and we have to transform all of them. We assume wlog that we arrange the variables by decreasing $v_i$, i.e. $v_i \geq v_2 \geq \ldots \geq v_n$. Then the contribution of an instance with $n$ such variables to the optimal, denoted by $\opt^{\text{small}}$, is 
\begin{align*}
   \opt^{\text{small}} \leq \E{}{\max_{i\in [n]}X_i} = \sum_{i=1}^n \lp( v_i p_i \prod_{j<i}(1-p_j) \rp) 
    \leq \frac{\e}{n} \cdot n 
    < 1.5\e \cdot \opt ,
\end{align*}
where the second inequality holds because the expression is maximized when $p_i$ tends to 1 for every $i$ (and we have also assumed that $v_i \cdot p_i \leq \frac{\e}{n}$) and the last one because we have normalized OPT (see Claim \ref{cl:bounded_opt_off}).

\end{proof}

\constloss*
\begin{proof}[Proof of Claim~\ref{cl:single-large-value}]
    We show the two inequalities separately. 
    \paragraph{For $\opt(X)\leq \opt(Y)$:} Assume that only one  variable in the instance falls in Case 2, thus it is the only one that we change. 

    \begin{itemize}
        \item \emph{Base case}: we have 2 variables $X_1, X_2$ in the original instance, and $X_1, Y_2$ in the changed instance. Since the order is random, there are two cases by coupling the arrival orders of the two instances.
        If $X_1$ arrives first, both $\opt(X)$ and $\opt(Y)$ see the same realization and the same threshold, and therefore have the same gain.
        If $Y_2$ (resp. $X_2$) arrives first then
        \begin{align*}
        \opt(Y) & = \E{}{\max(Y_2, \E{}{X_1})}  \\
        & = v'_2 p'_2 + (1-p'_2)\E{}{X_1}\\ 
        & \geq v_2 p_2 + (1-p_2)\E{}{X_1}\\ 
        & = \E{}{\max(X_2, \E{}{X_1})} \\
        & = \opt(X),
        \end{align*}
        where the inequality holds because from the transformation we always have $p'_2 \leq p_2$.
        \item \emph{Induction}: assume that the inequality holds for at most $n-1$ variables, and that the only changed variable is $Y_j$. Then, for $n$ variables we distinguish in the following cases: 
        \begin{itemize}
            \item Case 1: variable $Y_j$ (resp. $X_j$) arrives first. Then we have 
            \begin{align*}
\opt(Y)  & = v'_j p'_j + (1-p'_j) \opt(Y_{2:n}) \\
& \geq  v_j p_j + (1-p_j) \opt(X_{2:n}) \\ 
&= \opt(X),
            \end{align*} 
    where the inequality holds since $p'_j \leq p_j$ and from the induction hypothesis.
            \item Case 2: any of the non-changed variables $Y_k$ (resp. $X_k$) for some $k \neq j$ arrives first. Then, from the induction hypothesis we get
            \begin{align*}
\opt(Y)  & = \E{}{\max(Y_k, \opt(Y_{2:n})} \\
& \geq \E{}{\max(X_k, \opt(X_{2:n}))} \\ &= \opt(X).
            \end{align*}
        \end{itemize}
    \end{itemize}

    If there are $\ell \in [n]$ variables in Case 2, we can apply the transformations sequentially: At each time, pick at random a variable in Case 2 and transform it. Repeat until all $\ell$ random variables are transformed. By applying multiple times the arguments above, we know that we end up with the instance $Y$, for which it holds that $\opt(X) \leq \opt(Y)$.

    \paragraph{For $\opt(X) \geq \lp( 1-\e^2\rp)\opt(Y)$:} 
    To start, assume that we have again only one changed variable. We define an algorithm $\alg'$ as follows: $\alg'$ is given an instance $Z$ and, by coupling the realizations with instance $Y$, it simply always stops at the same index that $\opt$ stops at $Y$. 
    \footnote{Note that this is no longer a threshold algorithm, since the decision to stop in the instance $Z$ is based on the realizations of another instance $Y$ i.e. we might have ($Z<$threshold) but still stop.}
    Observe that $\alg'(Y) = \opt(Y)$ (since they make the exact same decisions) and that $\opt(X) \geq \alg'(X)$.

    We now define a new algorithm $\alg''$ for instance $X$: $\alg''$ stops on $X$ at the same index (and, thus, at the same variable when it is not the changed one) as $\alg'$ on instance $Y$, except in the case that $\alg'$ stops at the transformed variable $Y_i$. 
    In this case, $\alg''$ stops at the corresponding original variable $X_i$ whenever it is realized, that is, with an additional probability of $p_i - p'_i$ compared to $\alg'(Y)$. 
    Coupling the realizations of the two instances, observe that $\alg'(Y)$ and $\alg''(X)$ obtain the same reward when they stop with a random variable before the transformed one. 
    When they encounter the transformed variable $Y_i$ (resp. the original variable $X_i$), $\alg'(Y)$ and $\alg''(X)$ stop with different probabilities. But since the transformation is expectation-preserving, $\alg'$ and $\alg''$ get the same expected reward from the transformed and the original variable, respectively. But now with an additional probability of $p_i - p'_i$ we have that $\alg''(X)$ stops but $\alg'(Y)$ continues. Therefore, when this occurs, $\alg'(Y)$ receives an additional reward. Let us write $Y_{\text{rest}}$ for the remaining instance in this case. Then we have that 
%
% \begin{align*}
%     \opt(X) &\geq \E{}{\alg''(X)} \\ &= \E{}{\alg'(Y)} - \left(p_i - \frac{1}{e^n} \right) \E{}{\alg'(Y_{\text{rest}})}  \\
%     &\geq \E{}{\alg'(Y)} - \left(\e^2- \frac{1}{e^n}\right)\E{}{\alg'(Y)} \\
%     &= \left(1- \e^2 + \frac{1}{e^n} \right) \opt(Y),
%     \end{align*}
%  where the second inequality holds because variables in Case 2 have $p_i \leq \e^2$ and $\alg'(Y) \geq \alg'(Y_{\text{rest}})$. 
%  If we again have more than variables to transform, we apply the changes sequentially and we get at most $\left(1-  \e^2 + \frac{1}{e^n} \right)^n$ loss in total. 
\begin{align*}
    \opt(X) &\geq \E{}{\alg''(X)} \\ &= \E{}{\alg'(Y)} - \left(p_i - p'_i \right) \E{}{\alg'(Y_{\text{rest}})}  \\
    &\geq \E{}{\alg'(Y)} - p_i\E{}{\alg'(Y)} \\
    &= \left(1- p_i \right) \opt(Y),
    \end{align*}
 where the second inequality holds since $\alg'(Y) \geq \alg'(Y_{\text{rest}})$. 
 If we again have more variables to transform, we apply the changes sequentially and we get the 
 \[\opt(X)\geq \prod_i(1-p_i)\opt(Y)\] 
 where the product is taken over all variables with support value $>\frac{1}{\eps^2}$. Observe that the probability that at least one of the variables get realized is $1-\prod_i(1-p_i)\leq\eps^2$, otherwise the gambler can obtain at least 1 reward in expectation from these variables with support value $>\frac{1}{\eps^2}$. Thus $\opt(X)\geq (1-\eps^2)\opt(Y)$.
\end{proof}

%%%%%% Main QPtas theorem cont'd proof %%%%%%
\qptasMain*
\begin{proof}[Continued proof of Theorem~\ref{thm:QPTAS_2_point}]
We conclude the proof of Theorem~\ref{thm:QPTAS_2_point} by showing how to bound the error for the constant size support case.
\paragraph{Bounding the error:} 
We demonstrate here how the proofs of the arguments for the 2-point distributions can be adapted in order to give a similar error for the constant size support for each case. For the preprocessing note that the difference now is that for every variable $X_i$, for every point $v_i \in \text{supp}\lp( \dist_i \rp)$ that satisfies either $v_i \cdot p_i< \frac{\eps}{cn}$ or $v_i<\e$, we move all their mass to value $0$. In the case of variables with points of value $v_i < \e$, we will still lose in the worst case $1.5 \e \opt$ in our single-choice objective by removing them all. For the variables with points that satisfy $v_i p_i < \frac{\eps}{cn}$ the following change occurs in the proof of Claim \ref{cl:small_mean_vanishing_gain}: instead of now summing over $n$ points (one per random variable), we will now have to consider at most $n \cdot c$ such points in the summation. But since we scaled down the bound of $v_ip_i$ by a factor of $c$ we still have a multiplicative loss of $(1-1.5c\eps)$. For the Case 1 variables, Lemmas~\ref{lem:value_error_propagation} and \ref{lem:no_propagation_probability_error} hold for general distributions; as a result, the multiplicative error here remains $(1+\e)^2$. 
For the points in Case 2, as Claim~\ref{cl:single-large-value} still holds, we do not need to make any changes to the case.
\end{proof}

\subsection{Proofs from Section~\ref{subsec:qptas_general}}\label{apn:qptas_general}

\qptasBundling*

\begin{proof}[Proof of Lemma~\ref{lem:bundling}]
We show that if one variable changes, then the optimal thresholds remain the same, and the lemma follows by applying this process as many times as the number of variables.

Assume only one variable changes, we show the lemma using induction in the total number of variables of the instance. For the base case, there is only one variable remaining, $X$ in the original instance and $X'$ in the transformed one. Thresholds therefore are $\theta_1 = \E{\dist}{X} = \E{\dist'}{X'}=\theta'_1$. 

We now assume the lemma holds for at most $n-1$ variables, and there are $n$ variables $X_1, \ldots, X_n$, where the transformed instance has a variable $X_1'$ transformed according to the bundling~\ref{eq:qptas_general_bundling} and denote by $\dist'$ the product distribution of all initial variables and $X'_1$.  Then we have the following
\begin{align*}
\theta'_{1:n} & = \E{\dist'}{\frac{1}{n} \lp( \opt(X'_1|X_{-1}) + \sum_{i=2}^n \opt(X_i| X_{-i}) \rp) } &\text{Definition of }\theta’ \\
& = \frac{1}{n}  \lp(  \E{\dist'}{ \max(X'_1, \theta_{-1}) } + \E{\dist}{\sum_{i=2}^n \opt(X_i| \theta'_{-i})  } \rp) & \\
& = \frac{1}{n}  \lp( \vphantom{\sum_i^n}  \E{\dist'}{\max(X'_1, \theta_{-1}) | X'_1<1}\Pr{}{X'_1<1} \rp.& \\
& \lp. \hspace{2cm} + \E{\dist'}{ X'_1| X'_1>1}\Pr{}{X'_1>1} + \E{\dist}{\sum_{i=2}^n \opt(X_i| \theta_{-i})  } \rp) & \text{Ind. hyp. \& Fact~\ref{cl:algo_thresholds}}\\
& = \frac{1}{n} \lp( \vphantom{\sum_i^n}   \E{\dist}{\max(X_1, \theta_{-1}) | X_1<1}\Pr{}{X_1<1} \rp. & \\
& \lp. \hspace{2cm} + \E{\dist}{ X_1| X_1>1}\Pr{}{X_1>1}  + \E{\dist}{\sum_{i=2}^n \opt(X_i| \theta_{-i})  }  \rp) & \text{Definition of bundling}\\
& = \theta_{1:n}
\end{align*}

\end{proof}

\section{Proofs from Section~\ref{sec:ptas}}\label{apn:ptas}
\subsection{Proofs from Section~\ref{subsec: ptas binary small value}} \label{apn:ptas_binary}

\qvLb*
\begin{proof}
    Notice that the probability that at least one of the variables in $\Xssv$ gets realized with value $v$ is at most $1-e^{-\Qsvo}<\Qsvo$. Thus in the optimal solution, the total contribution from the variables in $\Xssv$ with value $v$ is:
    \begin{itemize}
        \item At most $\eps^{10}v\leq\eps^8$ for $v\leq \eps^{-2}$.
        \item At most $\Qsvo\cdot v<\eps^{10}\cdot Q_{max}\cdot v$ for $v\geq \eps^{-2}$. Note that since $\opt\leq 1$, it has to be that $Q_{max}\cdot v_{max}\leq 1$. Therefore the total loss in this case is at most $\eps^{10}$.
    \end{itemize}
   Therefore, removing those variables lead to at most $\eps^{8}$ loss in total reward.
\end{proof}

\skippingInitial*
\begin{proof}[Proof of Observation~\ref{skipping initial}]
   If $\Qsv=\Qsvo$, then nothing is changed. Otherwise, note that the probability that $v$ is realized in the remaining variables is at least
   \begin{equation*}
        1-\prod_{X_j\textrm{ in the remaining of }\Xsv^*}(1-p_i) = 1-\exp\left(-\sum_{X_j\textrm{ in the remaining of }\Xsv^*} q_j \right)\geq 1-e^{-1/\eps}>1-\eps.
   \end{equation*}
   Therefore even after ignoring $X_i$, the probability that $v$ is realized is at least $1-\eps$. Thus, for any strategy that wants to accept $X_i=v$, the gambler can instead ignore it and accept the first realized $v$ later. The gambler fails to accept the same reward $v$ with probability at most $\eps$.

\end{proof}

\smalErrSamples*
\begin{proof}[Proof of Claim~\ref{lem:smalErrSamples}]
By Chebyshev's inequality we have, 
\[\Pr{}{|Y-\E{}{Y}|\geq \eps\delta A} \leq \frac{\text{Var} [Y]}{(\eps\delta A)^2}.\]
Before bounding $\text{Var} [Y]$, we first bound a few useful terms.
\[\E{}{Y}=\eps\E{}{X}=\delta A\]
\begin{eqnarray*}
\E{}{\sum_{i<j}2Y_iY_j}
&=&m(m-1)\E{}{Y_1Y_2}
=\frac{m(m-1)}{n(n-1)}\sum_{i<j}2A_iA_j\\
&=&\frac{m(m-1)}{n(n-1)}\left( \lp(\sum_i A_i\rp)^2-\sum_i A_i^2\right)
<\delta^2\lp(A^2-\sum_iA_i^2\rp)
\end{eqnarray*}
Now we bound $\text{Var} [Y]$. By the definition of variance, 
\begin{eqnarray*}
\text{Var} [Y]&=&\E{}{Y^2}-(\E{}{Y})^2\\
&=&\mathbb{E}\lp(\sum_i Y_i\rp)^2-\delta^2 A^2\\
&=&\sum_i \E{}{Y_i^2 } +2\sum_{1\leq i<j\leq m}\E{}{Y_iY_j}-\delta^2 A^2\\
&<&\delta\sum_{i\leq n} A_i^2 + \delta^2\lp(A^2-\sum_iA_i^2\rp) -\delta^2 A^2\\
&=&(\delta-\delta^2)\sum_{i\leq n} A_i^2\\
&<&\delta\sum_{i\leq n} A_i^2
\leq \delta\max_i A_i\sum_{i\leq n} A_i
\leq \delta^2\eps^{10}A^2.
\end{eqnarray*}
The lemma follows immediately from Chebyshev's inequality.
\end{proof}

\subsection{A PTAS for general distributions}\label{subsec:ptas_general}

%\evagerg{maybe we can give some intuition on the process of making all variables binary/constant support size and using the constant/binary support case to get the result}

When the support size of each random variable can be more than 1, suppose that each variable $X_i$ is drawn from a distribution, that for each possible support value $v$ (being a power of $(1+\eps)$ between $\eps$ and $1/\eps^2$ or $v_{\max}$), $X_i=v$ with probability $\pv_i$. Define $\qv_i=-\ln(1-\pv_i)$. We observe that when $\pv_i<\eps$, if $\eps$ is small enough, $\pv_i<\qv_i<\pv_i+(\pv_i)^2$. Thus $\qv_i$ is a $(1+\eps)$-approximation of $\pv_i$ when $\pv_i$ is small enough.
% \\
% \textbf{Definition for support value $v\leq 1/\eps^2$.}
\paragraph{First Preprocessing.}
For each support value $v\leq 1/\eps^2$, similar to the case with binary distribution define $\Xssv$ to be the set of variables with small probabilities on support $v$:
\[\Xs^*_v=\{X_i|X_i\in \Xsv,\ s.t.\ \qv_i<\eps^{20}\}.\]

Let $\Qsvo=\sum_{X_{i}\in \Xssv}\qv_i$ denote the sum of $\qv_i$ for all random variables with support values $v$ being realized with small probabilities. If $\Qsvo<\eps^{10}$, it is safe to ignore the support $v$ of all variables in $\Xssv$, as they contribute less than $\eps^4\opt$ in the optimal selection. To be more precise, for a variable $X_i\in\Xssv$, it is safe to replace $X_i$ with $X_i'$ which is defined by moving all probability mass on $v$ to $0$: $\Pr{}{X'_i=v}=0$, and $\Pr{}{X'_{i}=v'}=\Pr{}{X_i=v}$ for every $v'\neq v$. Thus later in the section, we assume that $\Qsvo>\eps^{10}$.

Let $Q^*_v=\min(\eps^{-1},\Qsvo)$. The same as in the binary distribution case, whenever the gambler encounters a variable $X_i\in \Xssv$, while the remaining variables in $\Xssv$ have a total $\qv_i$ more than $Q^*_v$, it is safe to ignore this variable $X_i$ if the realized value is $v$. We again define $\Qov=\eps^4\Qsv$, and $\deltav=\frac{\Qov}{\Qsvo}$.

For each variable $X_i$, define its ``special truncation'' $\Xl_i$ as follows: for each support value $v>0$, $\Pr{}{\Xl_i=v}=0$, if $\qv_i<\eps^{20}$; $\Pr{}{\Xl_i=v}=\pv_i$, if $\qv_i>\eps^{20}$. All the removed probability masses are moved to $q^0_i=\Pr{}{\Xl_i=0}$. In other words, the small probability part of the variable is truncated. As there are only $\tilde{O}(1/\eps)$ different possible discretized probabilities for each support value, and $\tilde{O}(1/\eps)$ distinct discretized support values, thus there are only $\poly(1/\eps)$ distinct truncated variables. 
% \\
% \textbf{Definition for support value $v=v_{\max}$.} 
For support value $v=v_{\max}$, similar to the binary distribution case we define the changes to the definitions above. 
% Let   $\Qmax=\sum_{i}\qv_i$. 
$\Xssv$ is defined by \(\{X_i|X_i\in \Xsv,\ s.t.\ \qv_i<\eps^{20}\Qmax\}\). Support $v$ of all variables in $\Xssv$ can be ignored if $\Qsvo<\eps^{10}\Qmax$. The ``special truncation'' $X_i^+$ removes the probability mass on $v$ for each variable $X_i$ with $\qv_i<\eps^{20}\Qmax$.

\paragraph{Second Preprocessing.} Now we are ready to propose variants of online selection games, which all have a similar optimal reward to the original prophet secretary game. 
Again, let $Y_1,Y_2,\ldots,Y_n$ be the random realization of all variables, with $Y_i$ being the $i$th random variable that arrives. 
For each support value $v>0$, partition all random variables to $\eps^{-4}+1$ blocks $\Bv_0, \Bv_{1}, \ldots, \Bv_{\eps^{-4}}$ such that each block contains a sequence of consecutive variables. 
In particular, $\Bv_0$ contains the first $n\lp(1-\frac{Q^*_v}{\Qsvo}\rp)=n-\deltav n/\eps^4$ arrived variables; each of $\Bv_{1},\ldots,\Bv_{\eps^{-4}}$ contains $\deltav n$ consecutive variables. 
We now redefine different versions of prophet secretary games in the binary distribution case to the general setting.

\begin{figure}[H]
    \centering
    \pgfmathsetmacro{\dist}{2.5}
\pgfmathsetmacro{\distV}{2}
\begin{tikzpicture}

	\node (g1) at (0,0){\gr{\LARGE $G_1$}};
 	\node (g2) at (\dist,0){\gr{\LARGE $G_2$}};
  	\node (g3) at (2*\dist,0){\gr{\LARGE $G_3$}};
   	\node (g3s) at (3*\dist,0){{\LARGE $G_3^*$}};
    \node (g4) at (4*\dist,0){\gr{\LARGE $G_4$}};
    \node (g5) at (5*\dist,0){\gr{\LARGE $G_5$}};

    %Draw all back arrows
    \draw[->,thick, gray] (g2) to [in=315,out=225] node[below] {$(1-\e)\leq $}(g1);
   \draw[->,thick, gray] (g3) to [in=315,out=225] node[below] {$(1-\e)\leq $}(g2);
      \draw[->,thick] (g3s) to [in=315,out=225] node[below] {$(1-\e)\leq $}(g3);
    \draw[->,thick] (g4) to [in=315,out=225] node[below] {$\leq $}(g3s);
   \draw[->,thick,gray] (g5) to [in=315,out=225] node[below] {$(1-\e)\leq $}(g4);

   %Draw all forward arrows
	\draw[->,thick,gray] (g1) to [in=135,out=45] node [above] {$\geq$} (g2);
    \draw[->,thick,gray] (g2) to [in=135,out=45] node [above] {$\geq$} (g3);
        \draw[->,thick] (g3) to [in=135,out=45] node [above] {$\geq$} (g3s);
    \draw[->,thick] (g3s) to [in=135,out=45] node [above] {$\geq 1/(1+\e)$} (g4);
    \draw[->,thick,gray] (g4) to [in=135,out=45] node [above] {$\geq$} (g5);

  %Draw theoremnodes
  \node[] (thmG12) at (0.5*\dist,0) {Lem~\ref{thm:equal-G1-G2-general}};

    \node[] (thmG23) at (1.5*\dist,0) {Lem~\ref{thm:equal-G2-G3-general}};

    \node[] (thmG34) at (2.5*\dist,0) {Lem~\ref{thm:equal-G3-G3s-general}};
    
    \node[] (thmG3s3) at (3.5*\dist,0) {Lem~\ref{thm:equal-G3-G4-general}};
    
    \node[] (thmG45) at (4.5*\dist,0) {Lem~\ref{thm:equal-G4-G5-general}};
  
\end{tikzpicture}
    \caption{Series of reductions we use to prove Theorem~\ref{DP-gen}. The games in grey are mostly the same as the binary case, while game $G^*_3$ is a new addition to handle the general case.}
    \label{fig:summary_ptas_general}
\end{figure}
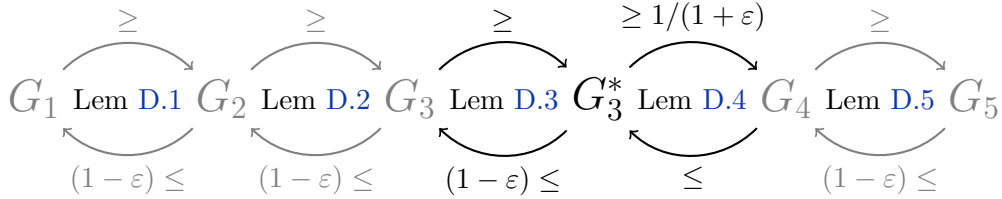

\paragraph{Original prophet secretary game $G_1$.} There are $n$ random variables $X_1,X_2,\ldots,X_n$ arriving in random order. At each time, when a variable arrives, the gambler observes the realized value and decides whether to take it. If the play decides not to take the realized value, the value cannot be taken later.

\bigskip
\noindent
\textbf{Prophet secretary game $G_2$ with possible failure.} The gambler plays prophet secretary game $G_1$, but the reward is set to 0 if there exists a possible support value $v$ and a block $\Bv_k$ such that for the variables $X_{i_1},X_{i_2},\ldots,X_{i_{\ell}}\in \Bv_k\cap \Xssv$, $\sum_{X_{i_j}\in \Bv_k}\qv_{i_j}\not\in[(1-\eps)\Qov,(1+\eps)\Qov]$. In other words, if the realization probability of $v$ of all variables with a small probability in a block is far from expectation, no reward is given to the gambler. The same as in the binary distribution case, we have the following lemma.

\begin{lemma}\label{thm:equal-G1-G2-general}
$(1-\eps)\opt(G_1)\leq \opt(G_2)\leq\opt(G_1)$.
\end{lemma}

\bigskip
\noindent
\textbf{Prophet secretary game $G_3$ with lazy decision.} The gambler plays prophet secretary game $G_2$. However, for every variable $Y_i$ in $\Xssv\cap \Bv_0$, whenever its realized value is $v$, the variable is removed and cannot be accepted by the gambler. Or equivalently, $Y_i$ is replaced by variable $Y_i'$ such that for every positive support value $v'\neq v$, $\Pr{}{Y_i'=v'}=\Pr{}{Y_i=v'}$; for support value $v$, $\Pr{}{Y'_i=v}=0$. The same as in the binary case, we have the following lemma.

\begin{lemma}\label{thm:equal-G2-G3-general}
$(1-\eps)\opt(G_2)\leq \opt(G_3)\leq\opt(G_2)$.
\end{lemma}

% \ytnote{Prove this Theorem~\ref{thm:equal-G2-G3-general}. Idea: Identical to Theorem~\ref{thm:equal-G2-G3-general}. May even skip the rigorous proof and say with words.}

Now we introduce a slightly different version of Game $G_3$ to reduce the correlation between the realization of each possible support value. In particular, for each $X_i$ with more than 1 positive support values, if support values $v_1,v_2,\ldots,v_k$ are not in the ``special truncation'' $\Xl_i$ of $X_i$, then $X_i$ is replaced with a sequence of random variables $\Xs_i=(X_{i,v_1}, X_{i,v_2}, \ldots, X_{i,v_k}, \Xl_i)$ that always come together one by one, where $X_{i,v_j}$ is a variable drawn from binary distribution: 
\[
X_{i,v_j} = \begin{cases}
    v_j, & \text{w.p. } p_i^{v_j}\\
    0,& \text{w.p. } 1-p_i^{v_j}
\end{cases}
\]
% with probability $p^{v_j}_i$, $X_{i,v_j}=v_j$, and with probability $1-p^{v_j}_i$, $X_{i,v_j}=0$. 
A nice property of the new game is that for each variable, the small-probability parts of the variable is separated from the large-probability part of the variable. The intuition behind this separation is that in later versions of the game, we can group all variables with binary distribution the same as in the previous section, and the remaining special variables can be treated separately via dynamic programs.

\bigskip
\noindent
\textbf{Perturbed prophet secretary game $G_3^*$ with lazy decision.} The gambler plays prophet secretary game $G_3$. However, when each variable $X_i$ arrives, it is replaced by a bundle of variables $\Xs_i$ that arrives consecutively.

\begin{lemma}\label{thm:equal-G3-G3s-general}
$(1-\eps)\opt(G_3)\leq \opt(G_3^*)\leq\opt(G_3)$.
\end{lemma}

\begin{proof}[Proof of Lemma~\ref{thm:equal-G3-G3s-general}]
Notice that $\max\{x\in \Xs_i\}$ is stochastically dominated by $X_i$: for any support value $v>0$, $\Pr{}{\max\{X\in \Xs_i\}=v}\leq \Pr{}{\exists X\in \Xs_i\textrm{ s.t. }X=v}=\Pr{}{X_i=v}$. This means that after replacing each $X_i$ with a bundle of variables $\Xs_i$, even if the gambler is allowed to take the maximum of all variables in $\Xs_i$ when $\Xs_i$ arrives, the optimal reward of the new game is upper bounded by the original game, thus $\opt(G_3^*)\leq\opt(G_3)$.

Now consider any optimal online algorithm for $G_3$. We assume that the algorithm is threshold-based, where the threshold at any step is determined by the optimal reward of the remaining game. At any node of the decision tree for $G_3$, there is a corresponding node in the decision tree for $G_3^*$ with the same history (i.e. whenever a variable $X_j$ arrives in $G_3$ at the $k$th step, a bundle of variables $\Xs_j$ arrives in $G_3^*$ at the $k$th step). Suppose that at a node of the decision tree for $G_3$, $X_i$ arrives, while the optimal algorithm sets the threshold to be $t$. If in the corresponding node in $G^*_3$ the gambler also sets the threshold to be $t$ at the corresponding node when $\Xs_i=(X_{i,v_1}, X_{i,v_2}, \ldots, X_{i,v_k}, \Xl_i)$ arrives, we show that such an algorithm in $G^{*}_3$ obtains at least $(1-\eps)$ fraction of the reward obtained in $\opt(G_3)$.

Firstly, the probability that the decision tree of $G^*_3$ reaches this node is larger than the probability that the decision tree of $G_3$ reaches the corresponding node, as in each of the previous steps in the decision trees, the probability of stopping in $G^*_3$ is always upper bounded by the probability of stopping in $G_3$, since both algorithms use the same threshold at each step, and $\max\{x\in \Xs_j\}$ is stochastically dominated by $X_j$ in each node of the tree. 

Secondly, conditioned on both decision trees reach the corresponding nodes, for every possible threshold $t$ and support value $v\geq t$, the probability that $G_3$ accepts $X_i$ with $X_i=v$ is $p_i^{v}$. We now analyze what's the probability that $G^*_3$ accepts $\Xs_i$ with value $v$. If $v$ is a support value of $\Xl_i$, then $v$ is accepted with probability $p_i^v$ at $\Xl_i$ if non of $X_{i,v_j}$ arrives, which happens with probability $1-\prod_j{(1-p_i^{v_j})}=1-e^{\sum_{i_j}q_i^{v_j}}>1-\eps$ when $\sum_{i_j}q_i^{v_j}<\tilde{O}(1/\eps)\eps^{20}<\eps^{18}.$ If $v=v_{j}$ is a support value with small probability in $X_i$, then $v$ is accepted with probability $p_i^v$ at $\Xl_i$ if non of $X_{i,v_\ell}$ (with $\ell\neq j$) arrives, which happens with probability at least $1-\prod_j{(1-p_i^{v_j})}>1-\eps$. Therefore, for every support value $v$, the probability that at a node of the decision tree of $G^*_3$ accepts a variable with value $v$, is at least $(1-\eps)$ times the probability that in the corresponding node of the decision tree of $G_3$ of the optimal algorithm accepts a variable with value $v$. Thus by coupling the randomness of the two decision trees, we have an algorithm with a reward at least $(1-\eps)\opt(G_3)$ in Game $G^*_3$, which means that $\opt(G_3^*)\geq(1-\eps)\opt(G_3)$. This concludes the proof of the lemma.
\end{proof}

\bigskip
\noindent
\textbf{Prophet secretary game $G_4$ with lazy decision and outside option.} The gambler plays prophet secretary game $G^*_3$. However, in each block $\Bv_k$ with $k\geq 1$, the realization of $v$ for variables with small probability in $\Xssv\cap \Bv_k$ is different. In particular, for all variables in $\Xssv\cap \Bv_k$, the realization of $v$ is put together to the beginning of $\Bv_k$ (and not later), and kept as an outside option $Z_{v}$. Whenever the gambler decides to accept a variable $X_i^+$ in $\Bv_k$, the gambler can choose to replace $X_i^+$ by $Z_{v}$. At the end of $\Bv_k$, the gambler can take $Z_v$ and stop. If the gambler decides to take nothing and continue the game, the outside option $Z_v$ is cleared.

\begin{restatable}{lemma}{ptasGeneralGThreestarGFour}\label{thm:equal-G3-G4-general}
$\opt(G^*_3)\leq \opt(G_4)\leq(1+\eps)\opt(G^*_3)$.
\end{restatable}

\begin{proof}[Proof of Lemma~\ref{thm:equal-G3-G4-general}]
After we have replaced each $X_i$ with a sequence of variables $\Xs_i$ where all support values with small probabilities are extracted as binary variables, the proof of Lemma~\ref{thm:equal-G3-G4} completely goes through in this case, as that proof does not require the special variables with large support probabilities to be binary, and also does not require the variables in each block $\Bv_k$ to arrive in random order. 
\end{proof}

\bigskip
\noindent
\textbf{Prophet secretary game $G_5$ with perturbed outside option.} The gambler plays prophet secretary game $G_4$. However, in each block $\Bv_k$ with $k\geq 1$, the outside option $Z_v$ is set to $v$ with a fixed probability $1-e^{(1-\eps)\Qov}$. The same as in the binary distribution setting, we have the following lemma.

\begin{restatable}{lemma}{GFourToGFive}\label{thm:equal-G4-G5-general}
$(1-\eps)\opt(G_4)\leq \opt(G_5)\leq\opt(G_4)$.
\end{restatable}

\begin{proof}
   Fix some value $v$. Note that the effect of the outside option on the reward is that we consider the maximum between that value and the value we get from the current variable. Even if in this maximum $v$ always wins, by making this change we have reduced this probability by at most $e^{-(1-\eps)\Qov}-e^{-(1+\eps)\Qov} = O(\eps^3)$. Moreover, there are at most $\tilde{O}(1/\eps)$ many different values. This means that the contribution of this change to $\opt(G_5)$ is at most $\eps\opt(G_4)$.

   Finally, comparing $G_3$ and $G_4$, we are reducing probabilities in the latter. So its optimal is always at most that of $G_3$. This concludes the proof. 

\end{proof}

Combining all the lemmas above, we get the following corollary, showing that Game $G_5$ obtains almost the same optimal reward as the original prophet secretary problem.

\begin{corollary}\label{thm:equal-G1-G5-gen}
$(1-O(\eps))\opt(G_1)\leq \opt(G_5)\leq(1+O(\eps))\opt(G_1)$.
\end{corollary}

To approximately solve $\opt(G_1)$, it suffices to solve $\opt(G_5)$. In the following theorem, we show that $\opt(G_5)$ can be efficiently computed.

\begin{theorem}\label{DP-gen}
    Prophet secretary game $G_5$ can be optimally solved in time $n^{(1/\eps)^{\poly(1/\eps)}}$.
\end{theorem}
\begin{proof}
 In the binary distribution case, we proposed a dynamic program that runs in $n^{\poly(1/\eps)}$ time. The same dynamic program still works for $G_5$ for general distribution, and the only difference is the size of the DP. Notice that previously in the binary case there are at most 2 support values for each variable, and now in the general case each variable has support size $\tilde{O}(1/\eps)$. This means that the number of distinct ``special variables'' with large probabilities for each support value is $\tilde{O}{(1/\eps)}^{\tilde{O}(1/\eps)}$, since each support value has $\tilde{O}(1/\eps)$ different discretized probabilities. The running time of the dynamic program is now $n^{(1/\eps)^{\poly(1/\eps)}}$, which means that the dynamic program is still a PTAS.
\end{proof}

\begin{corollary}
There exists a PTAS algorithm for the prophet secretary game $G_1$. This is a direct result of  \Cref{thm:equal-G1-G5-gen} and \Cref{DP-gen}.
\end{corollary}

\begin{table}[H]
\begin{tabular}{c|c|m{5.5cm}}
\textbf{Symbol} & \textbf{Definition} & \textbf{Explanation} \\ \hline
$\Xs_v$    &   $\{X_i: supp(X_i)\in \{0,v\}\}$    &  Variables with support $0$ and $v$.\\\hline
$\Xs^*_v$  & $\{X_i|X_i\in \Xs_v,\ s.t.\ q_i<\eps^{20}\}$ for $v\leq 1/\eps^2$    &   Variables with support $0$ and $v$ and $q_i<\eps^{20}$       \\\hline
$\Xs^*_{v_{\max}}$  & $\{X_i|X_i\in \Xs_{v_{\max}},\ s.t.\ q_i<\eps^{20}\sum_{X_j\in \Xs_{v_{\max}}}q_j\}$    &   Variables with support $0$ and $v_{\max}$ and $q_i<\eps^{20}\sum_{X_j\in \Xs_{v_{\max}}}q_j$       \\\hline
$\Xslv$  &    $\Xslv=\Xsv\setminus\Xssv$    &    Special variables with support $v$ \\\hline
$\Xsl$   &    $\Xsl=\bigcup_{v}\Xslv$       &    All special variables                  \\\hline
$\Qsvo$    &   $\Qsvo=\sum_{X_{i}\in \Xssv}q_i$    &   Total mass of small-probability variables with support $0$ and $v$ \\\hline
 $\Qsv$               &   $\Qsv=\min(\Qsvo,1/\eps)$                &     Total mass of small-probability variables with support $0$ and $v$, capped at $1/\eps$.   \\ \hline
  $\Qov$            &  $\eps^4\Qsv$                 &   The average mass of small probability variables in each block      \\ \hline
  $\deltav$         &   $\Qov/ \Qsvo$                &  The number of variables in each block $\Bv_k$ is $\delta_v$ fraction of total variables   \\ \hline    
 $\Bv_i$ &   & $i$'th block of the partition of all variables defined by support value $v$. Different support values create different partitions.\\ \hline
$Z_v$ & & Outside option for a block with support $v$\\ \hline 
\end{tabular}
\caption{Notation cheatsheet for binary distribution case.}
\label{table:cheatsheet}
\end{table}

\end{document}